\documentclass[12pt]{article}
\usepackage{authblk}
\usepackage{geometry}
\usepackage{amsmath}
\usepackage{amssymb}
\usepackage{amsthm}
\usepackage{mathrsfs}
\usepackage{subfigure}
\usepackage{customcommands}
\usepackage{bm}
\usepackage{Nettastyle}
\usepackage{dsfont}
\usepackage{comment}
\usepackage[utf8]{inputenc}
\usepackage{calc}
\usepackage{accents}
\usepackage{amsmath}
\usepackage{amssymb}
\usepackage{amsfonts}
\newtheorem{definition}{Definition}
\usepackage{braket}
\usepackage[normalem]{ulem}
\usepackage{color}
\usepackage{bbold}
\usepackage{ulem}
\usepackage{mathtools}
\usepackage{tensor} 
\usepackage{tikz}
\usetikzlibrary{arrows, positioning, quotes}

\DeclareMathOperator{\tr}{tr}
\newcommand{\dbtilde}[1]{\accentset{\approx}{#1}}

\title{Twice Upon a Time: Timelike-Separated Quantum Extremal Surfaces}

\author[1]{Netta Engelhardt,}
\author[2]{Geoff Penington,}
\author[3]{and Arvin Shahbazi-Moghaddam}

\affiliation[1]{Center for Theoretical Physics, Massachusetts Institute of Technology, \\Cambridge, MA 02139, USA}
\affiliation[2]{Center for Theoretical Physics and Department of Physics,\\
University of California, Berkeley, CA 94720, U.S.A.}
\affiliation[3]{Stanford Institute for Theoretical Physics,\\ Stanford University, Stanford, CA 94305 USA}
\emailAdd{engeln@mit.edu}
\emailAdd{geoffp@berkeley.edu}
\emailAdd{arvinshm@gmail.com}

\abstract{The Python's Lunch conjecture for the complexity of bulk reconstruction involves two types of nonminimal quantum extremal surfaces (QESs): bulges and throats, which differ by their local properties. The conjecture relies on the connection between bulk spatial geometry and quantum codes: a constricting geometry from bulge to throat encodes the bulk state nonisometrically, and so requires an exponentially complex Grover search to decode. However, thus far, the Python's Lunch conjecture is only defined for spacetimes where all QESs are spacelike-separated from one another. Here we explicitly construct (time-reflection symmetric) spacetimes featuring both timelike-separated bulges and timelike-separated throats. Interestingly, all our examples also feature a third type of QES, locally resembling a de Sitter bifurcation surface, which we name a bounce. By analyzing the Hessian of generalized entropy at a QES, we argue that this classification into throats, bulges and bounces is exhaustive. We then propose an updated Python's Lunch conjecture that can accommodate general timelike-separated QESs and bounces. Notably, our proposal suggests that the gravitational analogue of a tensor network is not necessarily the time-reflection symmetric slice, even when one exists.}

\begin{document}
\maketitle
\section{Introduction}

Over the past decade, quantum extremal surfaces (QESs) have come to play a central role in our understanding of quantum gravity and holography. The QES prescription~\cite{RyuTak06, HubRan07, FauLew13, EngWal14} relates the von Neumann entropy of a boundary, or boundary subregion, $B$ in AdS/CFT to the generalized entropy
\begin{align} \label{eq:sgen}
    S_\mathrm{gen}(\gamma_\mathrm{min}) = \frac{\mathrm{Area}(\gamma)}{4G} + S_\mathrm{bulk}(\gamma_\mathrm{min})
\end{align}
of a surface $\gamma_\mathrm{min}$ homologous to $B$. Here $\gamma$ denotes a codimension-two bulk surface, $\mathrm{Area}(\gamma)$ is its classical area, $G$ is Newton's constant, and $S_\mathrm{bulk}(\gamma)$ is the von Neumann entropy of bulk quantum fields on one side of the surface $\gamma$. The particular surface $\gamma_\mathrm{min}$ that appears in \eqref{eq:sgen} is required to be quantum extremal --  i.e. a critical point of the functional \eqref{eq:sgen} with respect to small perturbations in the location of $\gamma_\mathrm{min}$. If multiple QESs (homologous to a given boundary region) exist, $\gamma_\mathrm{min}$ is the minimal S$_{\rm gen}$ such surface. A closely related idea is the notion of entanglement wedge reconstruction~\cite{AlmDon14, JafLew15, DonHar16, FauLew17, CotHay17}, or subregion-subregion duality, which roughly identifies the exterior of this minimal QES -- the so-called entanglement wedge -- with the bulk information that can be reconstructed from the state $\rho_{B}$ on $B$ and the algebra of operators on $B$~\cite{Har16}.\footnote{For subtleties and qualifications regarding the meaning and regime of validity of the QES prescription and entanglement wedge reconstruction, see \cite{HayPen18, AkeLei19, AkePen20, AkePen21, AkeLev23}.}

\begin{figure}
    \centering
\includegraphics[width=0.5\textwidth]{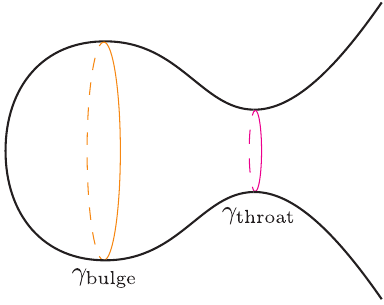}
    \caption{A spatial slice of a Python's lunch geometry. The boundary theory lives on a copy of  the asymptotic  boundary on the right. Figure reproduced from our previous work~\cite{EngPen21b}.}
    \label{fig:BulgevsThroat}
\end{figure}

In the last few years it has become clear that other ``nonminimal'' quantum extremal surfaces also play a crucial role in controlling the flow of information through the bulk-to-boundary map. In particular, the Python's lunch conjecture~\cite{BroGha19, EngPen21a, EngPen21b} proposes that the \emph{complexity} of reconstructing information within the entanglement wedge is exponential in the difference between the generalized entropies of two spacelike-separated \emph{nonminimal} QESs $\gamma_{\rm aptz}$ and $\gamma_{\rm main}$:
\begin{align} \label{eq:PLconjintro}
    C\sim \exp \left [\frac{1}{2} \left (S_{\rm gen}(\gamma_{\rm main})-S_{\rm gen}(\gamma_{\rm aptz}) \right) \right]
\end{align}
As shown schematically in Fig.~\ref{fig:BulgevsThroat}, these two QESs always have qualitatively different quasi-local properties. Specifically, the surface $\gamma_{\rm aptz}$, known as the appetizer surface, is always a local minimum of the generalized entropy on some partial Cauchy slice. We will more generally call any QES for which this quasi-local condition is true a ``throat''. The minimal QES is also always a throat -- indeed it is easy to construct continuous families of spacetimes where phase transitions exchange the minimal QES $\gamma_\mathrm{min}$ and the appetizer $\gamma_{\rm aptz}$. However, the QES $\gamma_{\rm main}$, which always has larger generalized entropy than $\gamma_{\rm aptz}$, is never a throat: there always exist local perturbations of $\gamma_{\rm main}$, within any Cauchy slice, which decrease its generalized entropy. Instead, it is a qualitatively distinct type of QES, that we call a bulge.

Much of the intuition behind the role of QESs in holography (and in particular for the Python's lunch conjecture) comes from tensor network toy models, which have provided valuable insights into AdS/CFT (see e.g.~\cite{Swi09, HaPPY, HayNez16}). In a tensor network version of a Python's lunch (shown in Fig.~\ref{fig:TNtoy}), the tensor network map from $\gamma_{\rm main}$ to $\gamma_\mathrm{aptz}$ features a large amount of postselection. As a result, the map from $\gamma_{\rm min}$ to $\gamma_{\rm aptz}$, which must be inverted in order to reconstruct the bulk behind $\gamma_{\rm aptz}$, is expected to have very high complexity. Indeed, the best known algorithm for such an inversion is Grover search~\cite{Gro96}, which has a complexity that is exponential in the amount of postselection. Since the amount of postselection in a tensor network is proportional  to $S_\mathrm{gen}(\gamma_{\rm main}) - S_\mathrm{gen}(\gamma_\mathrm{throat})$, \eqref{eq:PLconjintro} follows from the assumption that a) no algorithm faster than Grover search exists and b) that the gravitational bulk-to-boundary map has the same complexity as the analogous tensor network.

\begin{figure}
    \centering
\includegraphics[width=0.7\textwidth]{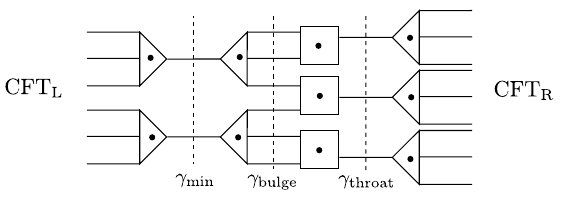}
    \caption{A tensor network representation of the python's lunch. The black dots are bulk legs, and the increase and decrease in the total bond dimension of the shown cuts introduces post-selection into the bulk-to-boundary map, from the bulk legs to the left of $\gamma_{\text{min}}$ to CFT$_R$. According to the Python's lunch conjecture, this results in an exponential enhancement in the map's complexity. Figure reproduced from our previous work~\cite{EngPen21b}.}
    \label{fig:TNtoy}
\end{figure}

An unfortunate feature of tensor networks as toy models of AdS/CFT is that they only describe an analogue of a single spatial slice rather than a fully covariant spacetime. Consequently, the emergence of bulk dynamical time in AdS/CFT remains deeply mysterious, with no clear information-theoretic interpretation. However, so long as all of the QESs homologous to a particular (complete) boundary $B$ are achronally-separated from one another, it is possible to find a Cauchy slice that contains all of those QESs simultaneously. Such Cauchy slices give a fairly direct analogue of a tensor network setup, as originally noted by~\cite{Swi09}. This achronal separation is often assumed in the literature -- implicitly or otherwise -- and occurs in many, or even most, situations of interest. In fact, to the authors' knowledge, no explicit examples of timelike-separated QESs homologous to the same boundary $B$ have so far appeared in the literature.\footnote{One construction of multiple \emph{classical} extremal surfaces anchored to a single boundary subregion that are not achronally separated from one another involves geodesics that wind around black holes in three spacetime dimensions~\cite{HubMax13}. However such surfaces are not individually achronal, which is a necessary condition for a surface to have well-defined generalized entropy, and hence to be potentially \emph{quantum} extremal.} However there is also no known general principle (such as the quantum focusing conjecture~\cite{BouFis15a}) that rules out such configurations.

In this paper, we demonstrate by explicit construction that timelike-separated QESs homologous to the same boundary connected component $B$ \emph{do} in fact exist. Our constructions of timelike-separated QESs have both the two qualitatively different types of QESs that appear in the Python's Lunch: ``throats'', which are locally minimally on some Cauchy slice and ``bulges'', which are not. Surprisingly, our constructions \textit{always} feature a third type of QES, which we call a ``bounce''. Heuristically, unlike both bulges and throats, bounces are locally minimal rather than maximal in time. Familiar examples of bounces are the bifurcation surface in de Sitter space and the inner bifurcation surface of a Kerr-Newman black hole (including the more symmetric situation of a Reissner-Nordstrom black hole). 
Bounces have so far not played a prominent role in AdS/CFT; e.g. the inner horizon of an AdS-Reissner-Nordstrom black hole is generally regarded as unphysical because it lies on a Cauchy horizon and is therefore not contained in the Wheeler-de Witt patch of any boundary slice. 

Since bounces appear in all of our explicit examples of timelike-separated bulges and throats, it is tempting to suggest that bounces are \emph{always} present in spacetimes with timelike-separated extremal surfaces. We show that this is true classically in spherically symmetric spacetimes; whether it is always true more generally is an interesting open question.

To be concrete, we construct time symmetric two-sided initial data for JT gravity minimally coupled to a massive scalar. Our initial data slices have three extremal surfaces: two throats on either side of a bounce (the slices are $\mathbb{Z}_{2}$-symmetric about the bounce) and asymptote to the usual two-sided black hole in pure JT. The resulting domain of dependence has two bulges or two throats in the past and future of the bounce depending on the choice of profile for the dilaton; the spacetime is illustrated in Fig.~\ref{fig:examplespacetime}. As semiclassical Lorentzian solutions, these spacetimes are well behaved.

\begin{figure}
    \centering
    \includegraphics[width=0.5\textwidth]{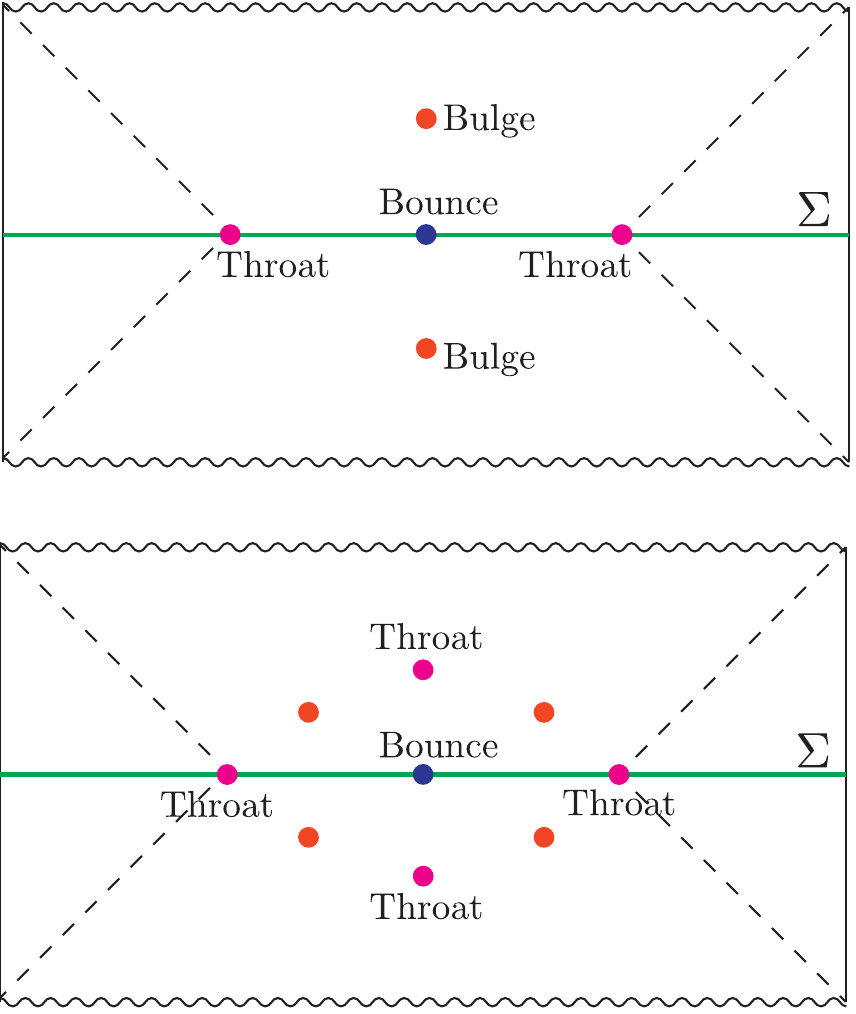}
    \caption{A schematic representation of our JT solutions with time-like separated extremal surfaces homologous to the right (or the left) boundary. Both solutions contain a time-symmetric slice $\Sigma$ with a bounce QES in the middle.}
    \label{fig:examplespacetime}
\end{figure}

The existence of such solutions poses a puzzle for our current, tensor-network-inspired understanding of the bulk-to-boundary map. In particular, the most general proposed version of the Python's lunch conjecture assumes that all of the QESs homologous to $B$ are spacelike-separated. We therefore propose a further refinement of the Python's lunch conjecture that is valid for completely general spacetimes, including those with timelike-separated QESs. This refined formula is compatible with earlier versions of the conjecture, and suggests necessary ingredients for constructing analogues of tensor networks in spacetimes where the natural tensor network picture fails to accurately represent bulk reconstruction.

An interesting feature of our refined conjecture is that, for spacetimes with time-reflection symmetry, the surfaces $\gamma_{\rm main}$ and $\gamma_\mathrm{aptz}$, which together compute the complexity of reconstruction, do not necessarily lie within the time-reflection symmetric slice. By contrast, the minimal QES $\gamma_\mathrm{min}$ does always lie within the time-reflection symmetric slice. It has traditionally been assumed that, for spacetimes with time reflection symmetry, the Cauchy slice most analogous to a tensor network is the time-reflection symmetric slice; our refined conjecture suggests that this is not the case.

The outline of the paper is as follows. In Sec.~\ref{sec:BBandT} we state our assumptions, give a rigorous definition of bounces, bulges, and throats, and prove a number of useful properties of these surfaces. We focus especially on a distinguished class of QESs, which we term `outer-minimal QESs'; these are particularly physically significant for motivating our newly refined Python's Lunch proposal. Sec.~\ref{sec:egs} presents our explicit examples of timelike-separated QESs. In Sec.~\ref{sec:proposal} we motivate and state our refined Python's Lunch proposal. The appendix is dedicated to technical details omitted from Sec.~\ref{sec:3-1}.

\section{Bounces, Bulges, and Throats} \label{sec:BBandT}

In this section, we give a natural classification of QESs into bounces, bulges, and throats based on their quasi-local properties. We then prove that a QES found by the maximin procedure is a throat while a QES found by the maximinimax procedure is a bulge. At the end of this section, we defined a distinguished class of throats -- outer-minimal QESs -- and discuss their relevant physical properties, which will turn out to be useful for the refined python's lunch conjecture in Sec.~\ref{sec:proposal}.

\paragraph{Assumptions:} Before we begin, let us first state state some basic assumptions about our setup and set some basic terminology and notation:
\begin{itemize}
\item The spacetime $\mathcal{M}$ is always asymptotically AdS and causally well-behaved. We also assume that the domain of dependence of any compact region is compact in time; this rules out e.g. de Sitter-type asymptotic infinities.

\item By a surface $\gamma$ we mean an achronal codimension-two embedded submanifold in the spacetime manifold $\mathcal{M}$. In this paper, we will further restrict surfaces to be homologous to some fixed boundary (partial) Cauchy slice $B$. Let $H$ be a spacelike homology slice of $\gamma$. Then the \emph{outer wedge}~\cite{EngWal17b} of $\gamma$, denoted by $\mathcal{W}_O[\gamma]$, is the domain of dependence $D(H)$.

\item Given a wedge $\mathcal{W} = D(H)$ for some partial Cauchy slice $H$, we denote by $\mathcal{W}'$ the \emph{complement wedge} defined by $D(H^c)$, where $H \cup H^c$ is a Cauchy slice for the entire spacetime. 

\item Given a surface $\gamma$, we define outwards (pointing towards $B$) orthogonal null vector fields $k^a$ and $\ell^a$ on it, which are respectively future and past directed and satisfy $k^a \ell_a = 1$. 

\item By a quantum extremal surface (QES) of $B$, we will mean a surface which is a stationary point of the generalized entropy functional 
\be\label{eq:sgen2}
\frac{\rm Area}{4 G} +S_{\rm ren}[\rho_{\rm out}]
\ee

and is homologous to $B$. Here $G$ is the renormalized Newton's constant and $S_{\rm ren}[\rho_{\rm out}]$ denotes the renormalized von Neumann entropy of the quantum fields on the partial Cauchy slice bounded by $\gamma$ and $B$.
\item Quantum Expansion: Given a surface $\gamma$, we choose coordinates $(u,v,y^i)$ on the normal bundle $N \gamma$ such that the induced metric on $N\gamma$ takes the form
\begin{align} \label{eq:metric}
    ds^2 = du dv + h_{ij}(y) dy^i dy^j
\end{align}
 where $y^i$ are coordinates on $\gamma$ and $k^\mu = (\partial_v)^{\mu}$ and $\ell^\mu = (\partial_u)^{\mu}$ are null vectors orthogonal to $\gamma$. The quantum expansions of the surface $\gamma$ are given by~\cite{BouFis15a}:\footnote{The original definition of the quantum expansion included additional factor of $4G$ for consistency with the classical expansions. We drop that here to make other formulas cleaner. Geometrically, the functional derivative of $S_\mathrm{gen}$ is a section of the conormal bundle $N^* \gamma$. $(\Theta_v, \Theta_u)$ is a section of $N \gamma$ that can be identified with this functional derivative via the inner product on sections of $N \gamma$ induced by \eqref{eq:metric}.}
 \begin{align}
\Theta_{(v)}(\gamma;y) = \frac{1}{\sqrt{h}} \frac{\delta S_{\text{gen}}}{\delta V(y)} \\
\Theta_{(u)}(\gamma;y) = \frac{1}{\sqrt{h}} \frac{\delta S_{\text{gen}}}{\delta U(y)}.
\end{align}
For an extremal surface, we have $\Theta_{(v)}(\gamma;y) = \Theta_{(u)}(\gamma;y) = 0$ everywhere.
 
 \item Small deformations of extremal surfaces: We can identify the bundle $N \gamma$ with a small neighbourhood of $\gamma$ via the exponential map $\mathrm{exp}_\varepsilon: N \gamma \to \mathcal{M}$. Given $p \in \gamma$ and orthogonal tangent vector $W \in N_p \gamma$, we have
 \begin{align}
  \mathrm{exp}_\varepsilon (p, W) =  \lambda_{p,W}(\varepsilon) 
 \end{align}
 where $\lambda_{p,W}(\varepsilon)$ is the unique affine geodesic satisfying $\lambda_{p,W}(0) = p$ and $\dot\lambda_{p,W}(0) = W$, and $\varepsilon$ is a formal small parameter.
 A local deformation of $\gamma$ can then be specified by a section $v=V(y^i)$ and  $u= U(y^i)$ of $N \gamma$. At first order in $\varepsilon$, the quantum expansions of the deformed surface are given by:
\begin{align}
\delta\Theta_{(v)}(U, V;y) =  \frac{\varepsilon}{\sqrt{h}} \left[\int dy'\frac{\delta S_{\text{gen}}}{\delta V(y)\delta V(y')} V(y') + \int dy'\frac{\delta S_{\text{gen}}}{\delta V(y)\delta U(y')} U(y')\right]\nonumber\\
\delta\Theta_{(u)}(U, V;y) =  \frac{\varepsilon}{\sqrt{h}}\left[ \int dy'\frac{\delta S_{\text{gen}}}{\delta U(y)\delta V(y')} V(y') + \int dy'\frac{\delta S_{\text{gen}}}{\delta U(y)\delta U(y')} U(y')\right].\label{eq:deltaTheta}\end{align}

\item We assume the quantum focusing conjecture (QFC)~\cite{BouFis15a}, which states that for any surface $\gamma$:
\begin{align}
\frac{\delta}{\delta V(y)} \Theta_{(v)}(U=0,V;y) \leq 0 \label{eq-QFCV}
\end{align}
The same relation holds with $V$ replaced by $U$.

Note that assuming the weaker, restricted quantum focusing conjecture~\cite{Sha22} which demands $\delta \Theta_{(v)}/\delta V \leq 0$ at points on the surface where $\Theta_{(v)}=0$ suffices for our purposes. This weaker conjecture was made and proved in holographic braneworld theories in~\cite{Sha22}.

\item Lastly, we assume a genericity condition: For any QES $\gamma$ considered in this paper, Eq. \eqref{eq-QFCV} for both U and V directions is not saturated everywhere on $\gamma$. Starting from any spacetime that does not satisfy the genericity condition on some QESs of interest, we expect that an arbitrarily small perturbation can induce genericity on the QESs.

\end{itemize}

We now proceed to give a precise definition of the three different classes of QESs. If $\gamma$ is a QES, then by definition, $\delta S_{\rm gen} / \delta V = \delta S_{\rm gen}/\delta U=0$ at $U=V=0$. To classify QESs, we therefore consider the second functional derivatives appearing in \eqref{eq:deltaTheta}. We will begin with the warm-up case of transverse symmetric QESs -- e.g. spherical symmetric QESs in spherically symmetric spacetimes.

\subsection{The transverse symmetric case}
By considering only transverse symmetric deformations we can construct a function $S_{\rm gen}(u,v)$ and effectively reduce the classification of QESs to its stationary points. Given a transverse-symmetric deformation
\begin{align}
W := \begin{pmatrix} u \\ v \end{pmatrix},
\end{align}
we can define a Hessian-type $2 \times 2$ matrix:
\begin{align}\label{eq-2dHessian}
\hat{L} =\left.\begin{pmatrix}
 \partial_u \partial_v S_{\rm gen} & \partial_v^2 S_{\rm gen}\\
\partial_u^2 S_{\rm gen} & \partial_u\partial_v S_{\rm gen}
\end{pmatrix}\right\rvert_{u=v=0}
\end{align}
such that
\begin{align}
 \begin{pmatrix} \delta \Theta_v \\ \delta\Theta_u \end{pmatrix} = \varepsilon \hat{L}  W,
\end{align}
The (pseudo-)inner product on spherically symmetric deformations induced by \eqref{eq:metric} is
\begin{align} \label{eq:syminner}
\langle W_1, W_2 \rangle :=  u_1^* v_2 +  v_1^* u_2.
\end{align}
In terms of this inner product, the change in generalized entropy created by the deformation is
\begin{align} \label{eq:symmdeltasgen}
\delta S_\mathrm{gen} = \frac{\varepsilon^2}{2} \langle W, \hat{L} W\rangle.
\end{align}

The matrix $\hat{L}$ always has two real eigenvalues; the smaller eigenvector is always spacelike, while the larger eigenvector is timelike. To see this, first note that, by the quantum focussing conjecture, $\hat{L}$ always has nonpositive off-diagonal terms, $\partial_v^2 S_{\rm gen} \leq 0$ and $\partial_u^2 S_{\rm gen}\leq 0$; generically these terms are strictly negative. We can therefore apply the Perron-Forbenius theorem to $\partial_u \partial_v S_{\rm gen} - \hat{L}$ to conclude that $\hat{L}$ has a spacelike real eigenvector with real eigenvalue $\lambda_1 \leq \partial_u \partial_v S_{\rm gen}$. The remaining eigenvalue must then be $\lambda_2 =2 \partial_u \partial_v S_{\rm gen} - \lambda_1 \geq \lambda_1$ so that $\tr[\hat{L}] = 2 \partial_u \partial_v S_{\rm gen}$. Since $\hat{L}$ and $\lambda_2$ are both real, the corresponding eigenvector must be real and hence, again by the Perron-Frobenius theorem, it must be timelike.

This result leads to the following natural definition of the three categories of QESs (under assumption of transverse symmetry): 
\begin{itemize} 
    \item A \emph{throat} is a QES with both positive eigenvalues;
    \item A \emph{bulge} is a QES such that the spacelike and timelike eigenvectors have negative and positive eigenvalues respectively;
    \item A \emph{bounce} is a QES with both negative eigenvalues.
\end{itemize}
Since we cannot have a spacelike and a timelike eigenvector with positive and negative eigenvalue respectively, this classification is exhaustive.\footnote{Here we are ignoring the possibility that one or both of the eigenvalues are zero.} In particular, throats are distinguished by the Hessian having positive smallest eigenvalue; as we shall describe below this distinguishing property has a particularly natural generalization to QESs without transverse symmetry. 

A spacelike eigenvector with positive eigenvalue means that perturbations in that direction increase $S_{\rm gen}$ at second order. By contrast, because timelike vectors have negative norm, a positive timelike eigenvector means that timelike perturbations \emph{decrease} $S_{\rm gen}$. As a result, throats are minima of $S_{\rm gen}$ in space and maxima in time; bulges are maxima in both space and time; and bounces are maxima in space and minima in time.

\subsection{The general case}

We would like to generalize the above categorization to general situations without transverse symmetry. This can be achieved with the following definitions:

\begin{definition}
A QES $\gamma$ is a $\emph{throat}$ if there exists a (partial) Cauchy slice containing $\gamma$ in its interior on which there are no surfaces homologous to $\gamma$ with smaller generalized entropy. 
\end{definition}

\begin{definition}
A QES $\gamma$ is a \emph{bounce} if, given any (partial) Cauchy slice $\Sigma$ containing $\gamma$, there exists a continuous 1-parameter family of (partial) Cauchy slices $\Sigma(\eta)$, where $\Sigma(\eta=0)=\Sigma$, with a continuous family of stationary surfaces $\gamma(\eta) \subset \Sigma(\eta)$ such that $\gamma(0) = \gamma$ and $S_{\text{gen}}(\gamma(\eta)) > S_{\text{gen}}(\gamma)$ for any $\eta \neq 0$.
\end{definition}

\begin{definition}
A \emph{bulge} is a QES which is neither a throat, nor a bounce.
\end{definition}

It is easy to see that these definitions reduce to the transverse symmetric definitions given above. However there is also an equivalent general definition of a throat QES that is a more direct generalization of the transverse symmetric definition. 

The natural generalization of \eqref{eq-2dHessian} is the linear operator that maps an infinitesimal deformation 
\begin{align}
W(y) := \begin{pmatrix} U(y) \\ V(y) \end{pmatrix}
\end{align} 
of the QES to the quantum expansions of the deformed surface:
\begin{align}\label{eq-Lhat}
\hat{L}_\gamma W(y) &= \frac{1}{\varepsilon}\begin{pmatrix} \delta\Theta_v(y) \\ \delta\Theta_u(y)) \end{pmatrix} \\[12pt]&= 
\begin{pmatrix}
\hat{D}_{+} U(y^i) +  \int dy' \frac{\delta^2 S_{\text{ren}}}{\delta U(y) \delta U(y')}\rvert_{\text{off-diag}}  U(y') + \int dy' \frac{\delta \Theta_v (y)}{ \delta V(y')}  V(y')\\[11pt]
 \hat{D}_{-} V(y) + \int dy' \frac{\delta^2 S_{\text{ren}}}{ \delta U(y) \delta V(y')}\rvert_{\text{off-diag}}  V(y') +  \int dy' \frac{\delta \Theta_u (y)}{ \delta U(y')}  U(y')
\end{pmatrix}\nonumber
\end{align}
where the derivatives are evaluated at $\gamma$ and $S_{\text{ren}}$ is the renormalized entropy. Here,
\begin{align}
    \frac{ \delta S_{\text{ren}}}{\delta U(y) \delta V(y')} = S_{vu}'' \delta^{d-1}(y-y')+ \left.\frac{\delta S_{\text{ren}}}{\delta U(y) \delta V(y')}\right\rvert_{\text{off-diag}}
\end{align}
where $S_{vu}$ is defined as the coefficient of the delta function piece of the LHS. Finally,
\begin{align}
    \hat{D}_{\pm}=\frac{1}{4G} \left[-\nabla^2 \mp \chi^i \nabla_i + R_{\mu\nu} k^\mu \ell^\nu -\frac{R}{2}\right]+ ~S_{vu}''
\end{align}
where $R_{\mu\nu}$ is the spacetime Ricci tensor and $\chi^i = k^\mu \nabla_i \ell_\mu$ is the twist. The classical limit, i.e. $G \to 0$, of $\hat{L}_\gamma$ was studied in~\cite{EngWal18} (See also~\cite{EngFis19, AndMar07}). The second line of \eqref{eq-Lhat} follows directly from plugging \eqref{eq:sgen2} into \eqref{eq:deltaTheta} and evaluating the classical piece explicitly. The quantum contributions to the semiclassical operator $\hat{L}_\gamma$ (henceforth referred to as the quantum stability operator) cannot be ignored as the quantum and classical terms can in principle be of the same order. By strong subadditivity, we have
\begin{align}\label{eq:ssaUV}
\frac{\delta^2 S_{\text{ren}}}{\delta V(y) U(y')}\rvert_{\text{off-diag}} \leq 0,
\end{align}
while
\begin{align} \label{eq:QFCUV}
\frac{\delta \Theta_v (y)}{\delta V(y')} \leq 0
\end{align}
by quantum focusing.

The (pseudo-)inner product on deformations $W(y)$ induced by \eqref{eq:metric} is
\begin{align} \label{eq:pseudoinnerproduct}
\langle W_1, W_2 \rangle := \int_\gamma dy \sqrt{h}\,\,\left[ U_1(y)^*V_2(y) +  V_1(y)^* U_2(y)\right].
\end{align}
The change in generalized entropy is given by
\begin{align} \label{eq:gendeltaSgen}
\delta S_\mathrm{gen} = \frac{\varepsilon^2}{2}\langle  W, \hat{L}_\gamma W \rangle_\gamma,
\end{align}
in close analogy with \eqref{eq:symmdeltasgen}.

$\hat{L}_\gamma$ satisfies the following theorem:
\begin{thm}
The operator $\hat{L}_\gamma$ in Eq. \eqref{eq-Lhat} with boundary conditions $\delta U\rvert_{\partial \gamma}= \delta V\rvert_{\partial \gamma}=0$  has a real eigenvalue $\lambda$ (called its principal eigenvalue) which is smaller than or equal to the real part of all other eigenvalues. Furthermore, the corresponding eigenvector $W(y) = (U(y), V(y))$,a vector field in the normal bundle of $\gamma$,  satisfies $U(y),V(y) \geq 0$ everywhere and hence describes an outwards achronal deformation of $\gamma$.
\end{thm}
\begin{proof}[Proof sketch]
Consider the operator $\hat{L}_\gamma + C$ for real constant
\begin{align}
  C > \sup_{\gamma} \left.(- 4G \int dy' \frac{\delta^2 S_{\text{ren}}}{\delta V(y) U(y')}\rvert_{\text{off-diag}} - \int dy' \frac{\delta \Theta_v (y)}{\delta V(y')}  -R_{\mu\nu} k^\mu \ell^\nu + R/2 - 4 G S_{uv}''\right).  
\end{align} The existence of a supremum is obvious for compact $\gamma$, and for non-compact $\gamma$, we expect that the asymptotic AdS boundary conditions guarantee its existence. As we are working at a physics level of rigor, we will assume (without proof) the existence, uniqueness and regularity of the solution to the equation
$$ (\hat L_\gamma + C) \begin{pmatrix} U(y) \\  V(y) \end{pmatrix} =  \begin{pmatrix}f(y) \\ g(y) \end{pmatrix}$$
with smooth functions $f(y), g(y)$ on $\gamma$ and boundary conditions $ U\rvert_{\partial \gamma} =  V\rvert_{\partial \gamma} = 0$.\footnote{In the absence of quantum terms, the above assumptions are standard facts about elliptic partial differential equations.} 

We first show that when $f(y), g(y) > 0$ we have $U(y), V(y) > 0$ everywhere on $\gamma-\partial \gamma$. We prove this by contradiction. Suppose $\inf_\gamma V \geq \inf_\gamma U = U_\mathrm{min} < 0$ (the proof for $\inf_\gamma V \leq \inf_\gamma U$ is identical). Since we have $U\rvert_{\partial \gamma} = V\rvert_{\partial\gamma} = 0$, the infinum must be achieved at some point  $y_\mathrm{min}$ in the interior of $\gamma$. We then have
\begin{align}
&0 < f(y_\mathrm{min}) = \hat{D}_{+} U\rvert_{y=y_\mathrm{min}} + 4G \int dy' \frac{\delta^2 S_{\text{ren}}}{\delta V(y_\mathrm{min}) U(y')}\rvert_{\text{off-diag}} U(y')  + \int dy' \frac{\delta \Theta_v (y_\mathrm{min})}{\delta V(y')} V(y') \nonumber\\
&+ C U(y_\mathrm{min})  \leq \hat{D}_{+} U(y_\mathrm{min}) +  U_\mathrm{min} \left[ 4G \int dy' \frac{\delta^2 S_{\text{ren}}}{\delta V(y_\mathrm{min}) U(y')}\rvert_{\text{off-diag}} + \int dy' \frac{\delta \Theta_v (y_\mathrm{min})}{\delta V(y')} + C\right]. \label{eq:contradiction}
\end{align}
In the second inequality we used \eqref{eq:ssaUV} and \eqref{eq:ssaUV}. Finally, we note that
\begin{align}
\hat{D}_{+} U\rvert_{y_\mathrm{min}} \leq \left[R_{\mu\nu} k^\mu \ell^\nu - \frac{R}{2} + 4 G S_{uv}'' \right] U(y_\mathrm{min})
\end{align}
since the first-derivative term in $\hat{D}_{+} U\rvert_{y_\mathrm{min}}$ must be zero and the second derivative term nonpositive at $y_\mathrm{min}$. Consequently, the right hand side of \eqref{eq:contradiction} for sufficiently large positive $C$ is negative, giving our desired contradiction.

From this, the result about the principal eigenvalue is obtained by applying the Krein-Rutman theorem (KR). KR states that a compact linear operator $T$ on a Banach space $X$ which maps any non-zero element of a closed cone $K \subset X$ (i.e., a topologically closed subset of $X$ closed under addition and multiplication by non-negative scalars) into the interior $K$, necessarily has a unique real positive eigenvalue, larger than the complex norm of any other eigenvalue, and whose corresponding eigenvector belongs to the interior of $K$. See~\cite{ProtterWeinberger} for a proof of KR.

Taking $K$ to be the space of pairs of positive functions $(f(y),g(y))$ (with appropriate smoothness conditions\footnote{We will not worry about exactly which Sobolev space is most appropriate to work with here.}), we see that $(\hat{L}_\gamma + C)^{-1}$ satisfies the conditions of KR.\footnote{The existence and compactness of $(\hat{L}_\gamma + C)^{-1}$ follows from the assumed existence, uniqueness and regularity of the solution $U, V$ for any $(f,g) \in L^2(\gamma) \oplus L^2(\gamma)$. For noncompact surfaces $\gamma$, you also have to worry about the details of the boundary conditions: asymptotically AdS boundary conditions do not spoil the compactness of $(\hat{L}_\gamma + C)^{-1}$ for the same reason that the Laplacian on hyperbolic space (unlike in flat space!) has discrete spectrum.}

Since $(\hat{L}_\gamma + C)^{-1}$ is compact, it (and hence also $\hat{L}_\gamma$) has completely discrete spectrum. Define the inner product\footnote{Unlike the pseudo-inner product $\langle W_1 ,  W_2\rangle_\gamma$ defined in \eqref{eq:pseudoinnerproduct}, $( W_1 ,  W_2)$ is a true, positive semi-definite inner product.}
\begin{align}
    (W_1 , W_2) := \int_\gamma dy \sqrt{h}\,\,\left[ U_1(y)^* U_2(y) + V_1(y)^* V_2(y)\right].
\end{align}
Since, for any $W$, the contribution to $(W , \hat L_\gamma W)$ from the second-derivative terms in $\hat L_\gamma$ is always nonnegative, the contribution from first derivative terms is always purely imaginary, and the zero-derivative terms are bounded, we find that $\mathrm{Re}[(W , \hat L_\gamma W)]/(W , W)$ is bounded from below. 

It follows that there exists an eigenvalue of $\hat{L}_\gamma$ with smallest real part; let $\lambda$ be that eigenvalue. Then, for sufficiently large $C$, $(\lambda + C)^{-1}$ is the eigenvalue of $(\hat{L}_\gamma + C)^{-1}$ with largest magnitude. KR then implies that $\lambda$ is real and that the corresponding eigenvector for $\hat L_\gamma$ is real and spacelike.\end{proof}

From here on, when discussing a particular QES $\gamma$, we will denote by $\lambda$ and $W$ its corresponding principal eigenvalue and eigenvector. 
\begin{cor}
    Assuming the genericity condition, the principal eigenvector $W = (U,V)$ satisfies $U(y),V(y) > 0$ everywhere in the interior of $\gamma$ and hence describes a spacelike deformation.
\end{cor}
\begin{proof}
    Suppose $U(y_0) = 0$ for some $y_0$ in the interior of $\gamma$ and there exists $y'$ such that $V(y') > 0$. We have
    \begin{align}
        0 &= \lambda U(y_0) = \Theta_v(y_0) \nonumber\\&= \nonumber\hat{D}_{+} U\rvert_{y=y_0} + 4G \int dy' \frac{\delta^2 S_{\text{ren}}}{\delta V(y_0) U(y')}\rvert_{\text{off-diag}} U(y')  + \int dy' \frac{\delta \Theta_v (y_0)}{\delta V(y')} V(y')
        \\&< \hat D_+ U|_{y = y_0},
    \end{align}
    where the inequality follows from \eqref{eq:ssaUV} and \eqref{eq:QFCUV} together with the fact that $U,V\geq 0$ everywhere. It is strict by the genericity condition and the assumption that $V(y') > 0$ for some $y'$. But
    \begin{align}
        \hat D_+ U|_{y = y_0} = -\nabla^2 U(y_0) \leq 0
    \end{align}
    since $U \geq 0$ and $U(y_0) = 0$. This gives our desired contradiction. An identical argument rules out the case where $V(y_0) = 0$ for some $y_0$ in the interior of $\gamma$ and $U(y') > 0$ for some $y'$. Since the principal eigenvector $W$ is by definition nonzero, it follows that $U(y), V(y) > 0$ everywhere in the interior of $\gamma$.
    
\end{proof}

\begin{thm}\label{thm11}
The quantum stability operator of a \emph{throat} $\gamma$ has a non-negative principal eigenvalue, i.e. $\lambda \geq 0$.
\end{thm}

\begin{proof}
Suppose $\lambda < 0$. Then, for small enough $\varepsilon$, $\mathrm{exp}_\varepsilon(W)$ is quantum anti-normal by \eqref{eq-Lhat} and $S_{\rm gen}(\mathrm{exp}_\varepsilon(W)) < S_{\rm gen}(\gamma)$ by \eqref{eq:gendeltaSgen}.  Let $\Sigma$ be a partial Cauchy slice containing $\gamma$. Let $\gamma' = \partial \mathcal{W}_{O}[\mathrm{exp}_\gamma(\epsilon X^a)] \cap \Sigma$. By the QFC, $S_{\rm gen}(\gamma') \leq S_{\rm gen}(\mathrm{exp}_\varepsilon(W))$: so $\gamma'$ is a surface in $\Sigma$ homologous to $\gamma$ with smaller $S_{\rm gen}$ than $\gamma$. It follows that whenever $\lambda < 0$ the QES $\gamma$ cannot be a throat.
\end{proof}
If the converse of the statement of this theorem were true, it would furnish an alternative definition of a throat. In fact, the converse indeed holds for a generic class of QESs which we term nondegenerate:

\begin{definition}
A QES $\gamma$ is called isolated if there exists a tubular neighborhood of $\gamma$ which does not completely contain another QES homologous to the same boundary region. We call an isolated QES with principal eigenvalue $\lambda \neq 0$ a nondegenerate QES.
\end{definition}

We expect that most QESs of interest are either nondegenerate or can be made nondegenerate with an arbitrarily small perturbation. An example of a QES which is not isolated is the bifurcation surface of pure de-Sitter spacetime, since continuous spatial rotations (which do not preserve the static patch) deform the QES into a neighboring QES. A small perturbation at the time-symmetric slice of global de-Sitter can break this symmetry and create a nondegenerate QES.

As part of the next subsection, we prove that a nondegenerate QES with $\lambda >0$ is a throat, providing an alternative definition of a (nondegenerate) throat.

\subsection{Maximin and Maximinimax}

Though QESs can be defined as stationary points of the $S_{\rm gen}$ functional, useful independent definitions, which make it easier to prove certain properties (e.g. existence) have been given in the literature. The definitions on which we focus are the maximin and maximinimax procedures~\cite{Wal12, MarWal19, AkeEng19b, BroGha19} which identify a certain surface through a search algorithm in a globally hyperbolic region of spacetime. Here we review the maximin and maximinimax procedures and prove that they find throats and bulges respectively.

\begin{definition}
(From ~\cite{Wal12, MarWal19, AkeEng19b}) Let $\gamma_1$ and $\gamma_2$ be spacelike-separated surfaces homologous to a boundary region $B$, and with non-positive null quantum expansions towards each other, we define another QES $\mathrm{maximin}(\gamma_1,\gamma_2)$ as follows: Let $\mathcal{W}$ indicate the wedge between $\gamma_1$ and $\gamma_2$ (i.e., the domain of dependence of a compact partial Cauchy slices with boundary $\gamma_1 \cup \gamma_2$). Then $\mathrm{maximin}(\gamma_1,\gamma_2)$ is the surface which achieves the following maximinimization:
\begin{align}\label{eq-maximin}
\max_{\{C\}} \min_{\{\gamma_C\}} S_{gen}(\gamma)
\end{align}
where $\{C\}$ indicates the set of all Cauchy slices of $\mathcal{W}$, and $\{\gamma_C\}$ indicates the set of all surfaces homologous to $\gamma_1$ (or equivalently, $\gamma_2$) on a given $C$. We further demand that the maximin surface satisfies the following stability condition: under any infinitesimal deformation of the surface $\gamma$ and a maximin slice $C \supset \gamma$ to a deformed surface $\gamma'$ and slice $C' \supset \gamma'$ such that $\gamma'$ remains stationary within $C'$, we have $S_{\rm gen}(\gamma') \leq S_{\rm gen}(\gamma)$.
A slight variant of this where instead of $\gamma_1$ and $\gamma_2$, we have a boundary region $B$ and a homologous surface $\gamma_1$ is defined in the obvious analogous way with $\mathcal{W} = \mathcal{W}_O[\gamma]$ and is denoted by $\text{maximin}~(\gamma_1, B)$.
\end{definition}
As in earlier work, we will assume without rigorous proof that stable maximin surfaces always exist. In~\cite{Wal12, MarWal19, BroGha19, AkeEng19b, Bousso:2021sji}, it was argued that the quantum expansion conditions on $\gamma_1$ and $\gamma_2$ ensure that $\mathrm{maximin}(\gamma_1, \gamma_2)$ is a QES contained in $\mathcal{W}$ (similarly, if the null quantum expansions of $\gamma_1$ towards $B$ is non-positive, then $\text{maximin}~(\gamma_1, B)$ is a QES).

For a thorough discussion of (quantum) maximin, including arguments in favor of the existence of maximin surfaces, see~\cite{Wal12, MarWal19, AkeEng19b}.

\begin{definition}\label{def:maximinimax}
(From \cite{BroGha19})
Given two spacelike-separated throats $\gamma_1$ and $\gamma_2$, homologous to the boundary region $B$, another QES homologous to B, denoted by $\mathrm{maximinimax}(\gamma_1, \gamma_2)$, located in the wedge $\mathcal{W}$ between them (i.e., the domain of dependence of a compact partial Cauchy slices with boundary $\gamma_1 \cup\gamma_2$) is defined as the surface achieving the following maximinimaximization: 
\begin{align}\label{eq-maximinimax}
\max_{\{C\}} \min_{\{f_C\}} \max_{0\leq\eta\leq 1} S_{gen}(f_C^{-1}(\eta))
\end{align}
where $\{C\}$ denotes the set of all partial Cauchy slices of $\mathcal{W}$, and $\{f_C\}$ indicates the set of all sweep-outs of C, i.e. smooth nondegenerate functions $f_{C}:C \to  [0,1]$, and such that $f_{C}(\gamma_1) = 0$ and $f_{C}(\gamma_2)=1$~\footnote{Intuitively, the level sets of $f$ define a foliation of $\Sigma$.} We further demand that the maximinimax surface satisfies the following stability condition: under any infinitesimal deformation of the surface $\gamma$ and a minimax slice $C \supset \gamma$ to a deformed surface $\gamma'$ and slice $C' \supset \gamma'$ such that $\gamma'$ remains stationary within $C'$, we have $S_{\rm gen}(\gamma') \leq S_{\rm gen}(\gamma)$.
\end{definition}
We will assume without rigorous proof that stable maximinimax surfaces always exist. In~\cite{BroGha19}, it was argued that a maximinimax  is a QES in the interior of $\mathcal{W}$. Note that the restriction to compact in time $D(H)$ for compact $H$ is important here, since otherwise the bulge may run off to an asymptotic region.

\begin{thm}\label{thm33}
A maximin QES is a throat.
\end{thm}

\begin{proof}
The minimization step in the quantum maximin procedure ensures that there exists a partial Cauchy slice on which $\gamma$ minimizes $S_{\text{gen}}$, i.e. $\gamma$ is a throat.
\end{proof}

We are now able to prove the theorem adertised in the previous subsection.

\begin{thm}\label{thm22}
A nondegenerate QES $\gamma$ with $\lambda>0$ is $\mathrm{maximin}(\gamma_1,\gamma_2)$ for some pair of surfaces $\gamma_1, \gamma_2$.
\end{thm}

\begin{proof}
Let $W$ be the principal eigenvector of $\hat{L}_\gamma$ and define $\gamma^+ = \mathrm{exp}_\varepsilon(W)$ and $\gamma^- = \mathrm{exp}_\varepsilon(-W)$. Since the principal eigenvector $W$ is spacelike, for small enough $\epsilon>0$ there exists a partial Cauchy slice $\tilde{\Sigma}$ containing $\gamma$ such that $\partial \tilde{\Sigma} =  \gamma^+ \cup \gamma^-$. $\lambda>0$ implies that $\Theta_v (\gamma^+) > 0$, $\Theta_u (\gamma^+) > 0, \Theta_k (\gamma^-) < 0$ and $\Theta_{\ell} (\gamma^-) < 0$. Therefore, a restricted maximin process in $D(\tilde{\Sigma})$ returns a QES $\gamma'$. Since $\epsilon$ can be made arbitrarily small, the assumption that $\gamma$ is isolated implies that $\gamma=\gamma'$. Therefore, $\gamma$ is maximin.
\end{proof}

When restricting to nondegenerate QESs, the results of theorems~\ref{thm11}, \ref{thm33}, and~\ref{thm22} can be summarized in the following triality:
\vspace{0.6cm}

\begin{tikzpicture}[
    auto,
    block/.style={
        rectangle,
        draw,
        align=center,
        rounded corners
    },
    line/.style={
        double,
        draw,
        -implies,
        double distance=3pt,
        thick
    }
]

\node[block] (ndt) at (-2,{sqrt(12)}) {Nondegenerate Throat};
\node[block] (ndq) at (6,{sqrt(12)}) {Nondegenerate QES\\with $\lambda > 0$};
\node[block] (mq) at (2,0) {Nondegenerate maximin QES};

\draw [line] (ndt) -- (ndq) node[midway,above,sloped] {Thm 2};
\draw [line] (ndq) -- (mq) node[midway,above,sloped] {Thm 4};
\draw [line] (mq) -- (ndt) node[midway,above,sloped] {Thm 3};

\end{tikzpicture}\\

We will now shift gears to discuss the properties of bulges and bounces.

\begin{thm}
    A nondegenerate maximinimax QES $\gamma$ is a bulge.
\end{thm}

\begin{proof} Let $\Sigma$ be a minimax slice of $\gamma$. We first show that $\gamma$ is not a throat. Suppose there exists a partial Cauchy slice $\tilde{\Sigma}$ on which $\gamma$ is a minimum $S_{\rm gen}$ surface. By deforming $\Sigma$ to locally align with $\tilde{\Sigma}$ in a neighborhood of $\gamma$, we conclude that the minimax surface on the deformed slice has a larger $S_{\rm gen}$ than $\gamma$, so $\gamma$ cannot be maximinimax. Therefore, $\gamma$ is not a throat.

We will now show that $\gamma$ is not a bounce. The defining property of a bounce would imply that a deformation of the minimax slice leads to a slice with a stationary surface $\gamma'$ such that $S_{\rm gen}(\gamma') > S_{\rm gen}(\gamma)$. This would contradict the stability requirement of a maximinimax surface (see Definition~\ref{def:maximinimax}). Therefore, $\gamma$ is not a bounce. 
We conclude that $\gamma$ is a bulge.
\end{proof}
We have seen that a nondegenerate throat is equivalent to a nondegenerate QES with $\lambda>0$. The following two theorems state the principal eigenvalue properties of bulges and bounces.

\begin{thm}
A nondegenerate bulge has a negative principle eigenvalue.
\end{thm}
\begin{proof}
A nondegenerate bulge is in particular not a nondegenerate throat. Therefore, by the triality established above, $\lambda <0$.
\end{proof}

A similar result can be proven for a nondegenerate bounce:

\begin{thm}A nondegenerate bounce $\gamma$ has a negative principle eigenvalue, i.e. $\lambda <0$.
\end{thm}

\begin{proof} Suppose $\lambda>0$. By the construction in the proof of Thm.~\ref{thm22}, we conclude that $\gamma$ is maximin. Suppose $\Sigma$ is a maximin slice. Then, by the defining property of a bounce, a small deformation of $\Sigma$, called $\Sigma'$, exists containing a stationary surface with $S_{\rm gen}$ larger than that of $\gamma$. Since this is the unique stationary surface on $\Sigma'$, this contradicts the stability condition of maximin. Therefore, $\lambda <0$.
\end{proof}

In particular, a nondegenerate bounce cannot be a throat by Theorem~\ref{thm11}. Considering the definitions of a throat, bulge, and a bounce, an obvious implication is the following:
\begin{cor}
A nondegenerate QES is (exclusively) either a throat, a bulge, or a bounce.
\end{cor}

\subsection{Outer-minimal QESs} \label{sec:outerminimal}

We have hithero focused on local properties of QESs. In this section, we identify a class of particularly significant QESs based on a \textit{global} property which will turn out to be critical for our discussion of reconstruction complexity in Sec.~\ref{sec:proposal}. While the focus on this particular class of QESs may not prima facie appear particularly well-motivated, `outer minimality' is a property that in an intuitive sense controls the constriction of information flow from bulk to boundary: we will indeed find that an analysis of these surfaces is a sine qua non for a reformulated Python's Lunch. We therefore dedicate some time to studying their properties prior to proposing the new Python's Lunch prescription. 

\begin{definition} A throat $\gamma$ is \emph{outer minimal} if there is no other QES with smaller generalized entropy that is contained in $W_O[\gamma]$.
\end{definition}

It is natural to ask how this definition is related to previous notions of quasi-minimal QESs. While this reformulation is more convenient for our purposes, it turns out that our outer-minimal QESs are in fact identical to~\cite{EngWal18}'s `minimar' surfaces and \cite{AkeLev23}'s `(vN-)accessible' surfaces when the latter two categories are quantum extremal.\footnote{At least one of the authors of~\cite{EngWal18} feels that our terminology, at least for the specific case when the minimar surface is quantum extremal, is an improvement.}

\begin{lem}
An outer-minimal QES $\gamma$ is a throat.
\end{lem}
\begin{proof}
Suppose $\gamma$ is not a throat. Then, by Thm.~\ref{thm33}, it cannot be maximin. Therefore, a restricted maximin procedure in $\mathcal{W}_O[\gamma]$ would find a QES with smaller $S_{\text{gen}}$ in the interior of $\mathcal{W}_O[\gamma]$.
\end{proof}

Examples of outer-minimal QESs include the minimal QES and the outermost QES $\gamma_{\rm outer}$~\footnote{The outermost QES is a QES which is fully contained in the outer wedge of every other QES in the entanglement wedge; such a QES always exists~\cite{EngPen21a}.}. We will now show that an outer-minimal QESs homologous to $B$ will necessarily be contained in the entanglement wedge of $B$, i.e., the outer wedge of the minimal QES homologous to $B$. First, a precursor Lemma follows:

\begin{lem}\label{lem:ssa} Let $\gamma_a$ and $\gamma_b$ be quantum extremal surfaces homologous to the same boundary region $B$. Then there exist QESs $\gamma_c \subseteq W_O[\gamma_a] \cap W_O[\gamma_b]$ and $\gamma_d$ such that $\gamma_a, \gamma_b \subseteq W_O[\gamma_d]$ such that
\begin{align}
    S_\mathrm{gen}(\gamma_c) + S_\mathrm{gen}(\gamma_d) \leq S_\mathrm{gen}(\gamma_a) + S_\mathrm{gen}(\gamma_b).
\end{align}
\end{lem}
\begin{proof}
Throughout this proof, we will relax the definition of a surface, allowing it to contain null segments. Suppose $\gamma_a$ and $\gamma_b$ can be contained in a single Cauchy slice $\Sigma$ and further that $\gamma_a \cap \gamma_b$ does not contain a co-dimension two spacelike piece.  Define partial Cauchy slices $A = \mathcal{W}_{O}[\gamma_a] \cap \Sigma$ and $B = \mathcal{W}_O[\gamma_b]\cap \Sigma$. Then, $C = A \cap B$ and $D = A \cup B$ are also partial Cauchy slices. Furthermore, it is easy to see that $S_\mathrm{gen}(\partial C) + S_\mathrm{gen}(\partial D) \leq S_\mathrm{gen}(\partial A) + S_\mathrm{gen}(\partial B)$. To see this, note that the area terms on both sides of the inequality agree (except that the left-hand side excludes pieces where $A$ and $B$ touch but do not intersect) and the entropy terms obey the inequality by strong subadditivity, i.e. $S(C) + S(D) \leq S(A) + S(B)$.

Now, consider the general case where $\gamma_a$ and $\gamma_b$ may not live on a single Cauchy slice. Let $\tilde{\gamma_a}$ be the surface obtained evolving the portions of $\gamma_a$ to the future (past) of $\gamma_b$ along outward past (future)-directed null geodesics orthogonally fired from $\gamma_a$ until they intersect $\mathcal{W}_{O}[\gamma_b]$, and let $\tilde{\gamma_b}$ be obtained by evolving the portions of $\gamma_b$ to the future (past) of $\gamma_a$ along inward future (past)-directed null geodesics orthogonally fired from $\gamma_b$ until they intersect $\mathcal{W}_{O}[\gamma_a]'$. By the QFC, $S_\mathrm{gen}(\tilde{\gamma_a}) + S_\mathrm{gen}(\tilde{\gamma_b}) \leq S_\mathrm{gen}(\gamma_a) + S_\mathrm{gen}(\gamma_b)$. By construction, $\mathcal{W}_{O}[\tilde{\gamma_a}] \cap \mathcal{W}_{O}[\tilde{\gamma_b}]=\mathcal{W}_{O}[\gamma_a] \cap \mathcal{W}_{O}[\gamma_b]$ and $(\mathcal{W}_{O}[\tilde{\gamma_a}]' \cap \mathcal{W}_{O}[\tilde{\gamma_b}]')'=(\mathcal{W}_{O}[\gamma_a]' \cap \mathcal{W}_{O}[\gamma_b]')'$, but there exists a Cauchy slice $\Sigma$ that contains both $\tilde{\gamma_a}$ and $\tilde{\gamma_b}$. Therefore, $S_\mathrm{gen}(\partial C) + S_\mathrm{gen}(\partial D) \leq S_\mathrm{gen}(\gamma_a) + S_\mathrm{gen}(\gamma_b)$ where $C = A \cap B$, $D = A \cup B$, with $A = \mathcal{W}_{O}[\tilde{\gamma_a}] \cap \Sigma$ and $B = \mathcal{W}_O[\tilde{\gamma_b}]\cap \Sigma$. By construction, $\partial C$ and $\partial D$ are quantum anti-normal and normal respectively. As proved in~\cite{EngPen21a, EngPen21b}, this implies the existence of a QES $\gamma_c \subset \mathcal{W}_{O}[\gamma_a] \cap \mathcal{W}_{O}[\gamma_b]$ and $\gamma_d \subset (\mathcal{W}_{O}[\gamma_a]' \cap \mathcal{W}_{O}[\gamma_b]')'$ such that $S_\mathrm{gen}(\gamma_c) \leq S_\mathrm{gen}(\partial C)$ and $S_\mathrm{gen}(\gamma_d) \leq S_\mathrm{gen}(\partial D)$. Putting everything together, we have $S_\mathrm{gen}(\gamma_c) + S_\mathrm{gen}(\gamma_d) \leq S_\mathrm{gen}(\gamma_a) + S_\mathrm{gen}(\gamma_b)$.
\end{proof}

\begin{thm} \label{thm:outermininEW} Let $\gamma_a$ be an outer-minimal QES homologous to some boundary region $B$. Then $\gamma_a$ lives in the entanglement wedge of $B$.
\end{thm}

\begin{proof}
Let $\gamma_b$ be the minimal QES and let $\gamma_c$ and $\gamma_d$ be defined according to the construction of Lemma~\ref{lem:ssa}. In particular, $\gamma_a$, $\gamma_b$, $\gamma_c$, and $\gamma_d$ are all homologous to $B$.  By the outer-minimal property of $\gamma_a$, we have $S_{\rm gen}(\gamma_c) > S_{\rm gen}(\gamma_a)$. By the inequlaity proved in Lemma~\ref{lem:ssa}, $S_{\rm gen}(\gamma_d) < S_{\rm gen}(\gamma_d)$ contradicting that $\gamma_b$ is the minimal QES.
\end{proof}

As we will see in Sec.~\ref{sec:egs}, QESs homologous to the same boundary region may be timelike-separated from each other. This makes it particularly interesting that outer-minimal QESs are necessarily not timelike-separated from the minimal QES. This relation between the outer-minimal and minimal QES motivates an alternate definition of an outer-minimal QES:

\begin{thm}\label{thm:outermindefviaEW}
Let $\gamma$ be a QES of a boundary region $B$ and $W_{O}[\gamma]$ the corresponding outer wedge. Then $\gamma$ is outer minimal if and only if for all semiclassical spacetimes $(M',g',\ket{\psi'})$ that are identical to $(M,g,\ket{\psi})$ on $W_{O}[\gamma]$, the entanglement wedge of $B$ contains $\gamma$.
\end{thm}

\begin{proof}
It follows directly from Theorem \ref{thm:outermininEW} that an outer-minimal QES $\gamma$ homologous to $B$ lies in the entanglement wedge of $B$ in $(M',g',\ket{\psi'})$. Conversely, suppose $\gamma$ is not outer minimal. We can construct a spacetime $(M',g',\ket{\psi'})$ by gluing $W_{O}[\gamma]$ to its CPT conjugate across $\gamma$~\cite{EngWal17b, EngWal18}. Let $\gamma'$ be the minimal S$_{\rm gen}$ surface in $W_{O}[\gamma]$. Then, $(M',g',\ket{\psi'})$ contains the CPT conjugate of $\gamma'$, called $\gamma''$ such that $\gamma \subset \mathcal{W}_{O}[\gamma'']$. By a slight change of $(M',g',\ket{\psi'})$ in $\mathcal{W}_{O}[\gamma]'$ we can arrange for a  $S_{\rm gen} (\gamma'')<S_{\rm gen} (\gamma')$, while maintaining that $\gamma \subset \mathcal{W}_{O}[\gamma'']$. It then cannot be the case that the entanglement wedge of $B$ contains $\gamma$.
\end{proof}

Theorem \ref{thm:outermindefviaEW} suggests that all operators in $\mathcal{W}_{O}[\gamma]$ can be reconstructed from $B$ by someone \emph{knowing only about $\mathcal{W}_{O}[\gamma]$} if and only if $\gamma$ is outer minimal. Indeed, the `only if' direction here follows directly from combining Theorem \ref{thm:outermindefviaEW} with standard facts about entanglement wedge reconstruction~\cite{DonHar16}: we cannot reconstruct operators knowing only about $\mathcal{W}_{O}[\gamma]$ if those operators are not even guaranteed to be in the entanglement wedge.
On the other hand, to show the `only if' direction, and prove that such reconstructions are possible whenever $\gamma$ is outerminimal, we will need to take advantage of the following boundary density matrix associated to $\gamma$.

\begin{definition}Given an outer-minimal QES $\gamma$ homologous to boundary region $B$, define the boundary density matrix $\rho^{\gamma}_B$ associated to it via the CPT-conjugation protocol of~\cite{EngWal17b}: Let $\tilde{\mathcal{W}}_O [\gamma]$ be the CPT-conjugate of $\mathcal{W}_O [\gamma]$. A Cauchy slice of each wedge can be glued at $\gamma$ to construct a complete Cauchy slice whose domain of dependence is the full asymptotically AdS spacetime with boundary which is two copies of $B$. The corresponding CFT state restricted to $B$, then defines $\rho^\gamma_B$.
\end{definition}

Crucially, the density matrix $\rho^{\gamma}_B$ depends only on the bulk state restricted to $\mathcal{W}_O (\gamma)$, and not on the bulk state, or even the spacetime geometry outside $\mathcal{W}_O (\gamma)$. In the special case where $\gamma$ is the minimal QES, $\rho^{\gamma}_B$ is simply the boundary density matrix in the existing state. Then standard boundary reconstruction formulae~\cite{FauLew17, CotHay17, ChePen19, PenShe19} can be used to construct boundary duals of local bulk operators in the entanglement wedge of $B$, i.e. $\mathcal{W}_O[\gamma]$. It is an easy exercise to show that when $\gamma$ is a more general outer-minimal QES, if we replace the boundary density matrix with $\rho^\gamma_B$ in e.g. modular flow reconstructions ~\cite{FauLew17} or Petz map reconstructions \cite{CotHay17, ChePen19, PenShe19}, we can construct boundary duals to bulk operators localized to $\mathcal{W}_O[\gamma]$ that depend only on $\mathcal{W}_O[\gamma]$.

\section{Construction of Timelike-Separated Extremal Surfaces} \label{sec:egs}
In this section, we will explicitly construct spacetime solutions with timelike-separated bulges and throats in JT gravity with a minimally coupled massive scalar field. All of our examples feature a bounce between the timelike-separated bulges or throats. For simplicity, we work in the classical limit where, in particular, $S_{\rm gen}=A/4G$. 

\subsection{Time-like separated bulges}\label{sec:3-1}

In this subsection we construct time-like separated bulges. Though our explicit solutions are in JT gravity (or equivalently, in a higher dimensional near-extremal Reissner-Nordstrom black hole), the idea behind the construction is general and can in principle be implemented in other (e.g. higher dimensional) gravitating systems. We construct a single time-symmetric Cauchy slice $\Sigma$, with two asymptotically AdS boundaries, containing three extremal surfaces $\gamma_1$, $\gamma_2$, and $\gamma_3$ (ordered from left to right) homologous to one of the boundaries. Each boundary has a corresponding outermost extremal surface which (by a simple restricted maximin argument) always lies on the time-symmetric slice. Consequently, $\gamma_1$ and $\gamma_3$ are the outermost extremal surfaces of the left and right boundaries respectively (and are therefore throats). Let $H_{13}$ be the portion of $\Sigma$ between $\gamma_1$ and $\gamma_3$. Then, maximinimax within $D(H_{13})$ guarantees the existence of a bulge~\footnote{This assumes that maximinimax in $D(H_{13})$ does not run off to asymptotic infinity in the time directions. This could happen for instance in a spacetime where one glues a portion of de Sitter-Schwarzschild spacetime behind the horizons of two AdS-Schwarzschild spacetimes~\cite{FisMar14}. However, we expect (an expectation which is realized in our explicit construction below) a spacetime where D($H_{13}$) contains no asymptotic region.}. Now, suppose $\gamma_2$ is a bounce. Then, the $D(H_{13})$ maximinimax along with the time-symmetry of $\Sigma$ implies the existence of at least two bulge extremal surfaces which are timelike-separated from each other.

In the remainder of this subsection, we demonstrate an explicit example of such a $\Sigma$ in JT gravity minimally coupled to a massive scalar field. 

We work with the AdS$_2$ metric in the following coordinates (hereafter, we use units where the radius of curvature is unity)
\begin{align}
ds^2=-(1+x^2) dt^2 + (1+x^2)^{-1} dx^2.
\end{align}
The equations of motion for JT gravity coupled to a massive scalar field are given by
\begin{align}\label{eq-JTEOM}
-\nabla_\mu \nabla_\nu \phi + g_{\mu\nu} \nabla^2 \phi - g_{\mu\nu} \phi = \kappa~T_{\mu\nu},
\end{align}
where $\phi$ is the dilaton, $\kappa = 8\pi G$, and $T_{\mu\nu}$ is the matter stress energy tensor given by
\begin{align}
    T_{\mu\nu}  = \nabla_\mu \psi \nabla_\nu \psi - \frac{1}{2} g_{\mu\nu}(m^2 \psi^2 + \nabla_\sigma \psi \nabla^\sigma \psi),
\end{align}
and where the scalar field's equation of motion is
\begin{align}
\nabla^2 \psi + m^2 \psi = 0.
\end{align}
For simplicity, we will construct a solution which is reflection-symmetric across $\gamma_2$ (where we set $x=0$). Furthermore, we pick $\psi=0$ in the exteriors of the outermost extremal surfaces, i.e. to the left of $\gamma_1$ and to the right of $\gamma_3$. The task is then to construct a solution of $H_{13}$ which satisfies the constraint equation\footnote{Both the $tt$ and $tx$ components of Eq. \eqref{eq-JTEOM} are constraints, but the $tx$ component is automatically satisfied for time-symmetric initial data.}:
\begin{align}\label{eq-tt}
\phi - x \partial_x \phi - (1+x^2) \partial_x^2 \phi = \frac{\kappa}{2} \left(m^2 \psi^2 + (1+x^2) (\partial_x \psi)^2\right),
\end{align}
and such that $\phi_2$ is a bounce and $\phi\rvert_{\gamma_1} >0$ and $\psi\rvert_{\gamma_1}=0$ (by reflection symmetry, $\phi\rvert_{\gamma_3} >0$ and $\psi\rvert_{\gamma_3}=0$). Then, by gluing $\gamma_1$ and $\gamma_3$ to exteriors of a vacuum JT solution we can complete the initial data. The vacuum solution is
\begin{align}\label{eq-vac}
\phi = \phi_h \sqrt{1+x^2} \cos t,
\end{align}
where we set $\phi_h = \phi\rvert_{\gamma_1}$ and glue the $x>0$ part of the this solution to the right of $\gamma_3$ and the $x<0$ piece to the left of $\gamma_1$.
\begin{figure}[ht]
\begin{center}
\includegraphics[width=0.5\textwidth]{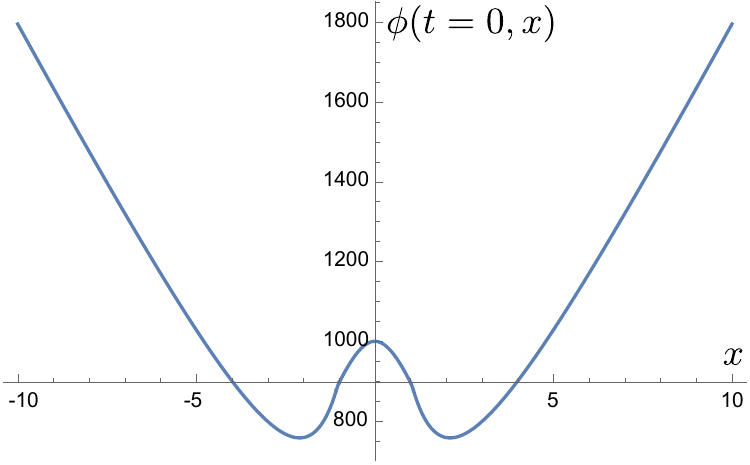}
\caption{The profile of $\phi$ on the slice $\Sigma$ ($t=0$ axis) shows three extremal surfaces, at $x=0$ and $x \approx \pm 2$.}\label{fig-Phi1}
\end{center}
\end{figure}

The only remaining challenge is to build a $H_{13}$ with a bounce. Here, a bounce means $\partial_t^2 \phi\rvert_{\gamma_2} >0$. From the $xx$ component of Eq. \eqref{eq-JTEOM}, we have
\begin{align}
\partial_t^2 \phi(t,x=0)\rvert_{t=0} = \frac{\kappa}{2} m^2 \psi(t=0,x=0)^2 - \phi(t=0,x=0).
\end{align}
We choose a profile of $\psi$ which is constant in an interval around $x=0$ and linearly decreases to zero at $\gamma_1$ and $\gamma_3$:
\begin{align}
\psi =     \begin{cases}
       \psi_0, |x| \leq x_0,\\
        \psi_0 -a |x-x_0|, x_0 <|x| \leq x_0+\frac{\psi_0}{a},\\
        0, x>x_1.
    \end{cases}
\end{align}
We can pick $\phi(t=0,x=0)$ freely, and let the constraint equation \eqref{eq-tt} fix $\phi(t=0,x)$. For the following choice of parameters, the solution will have the desired throats with an exterior given by the the vacuum solution \eqref{eq-vac}.
\begin{align}
\begin{cases}
\kappa=0.01,\\
m=100,\\
\psi_0=50,\\
\phi(t=0,x=0)=100,\\
a=600.
    \end{cases}
\end{align}
Figure~\ref{fig-Phi1} shows the resulting profiles of $\phi$ on $\Sigma$.

We can explicitly check the existence of the bulges by evolving the data in $D(H_{13})$. We do so by numerically solving for the solution off of the $t=0$ slice. Since the equations are linear, we can also solve analytically the solution for $\psi$ using the appropriate Green's functions, and use that to solve for the value of $\phi$ off of the $t=0$ slice. We do so in Appendix~\ref{app-1} and find perfect agreement with the results here.

\begin{figure}[ht]
\begin{center}
\includegraphics[width=0.5\textwidth]{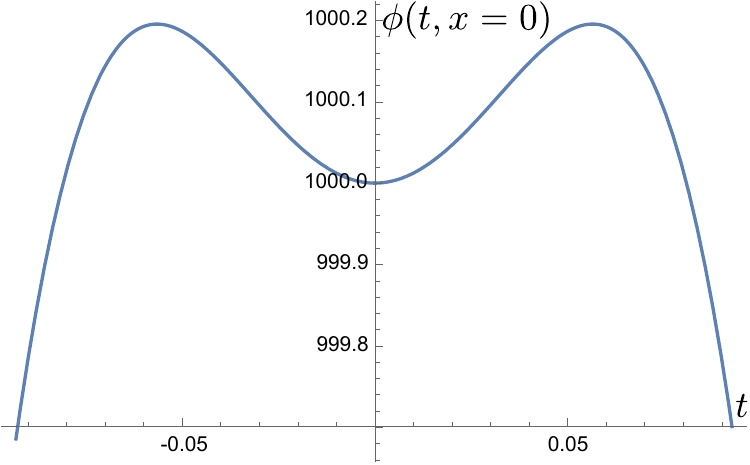}
\caption{By evolving the dilaton off of the time-symmetric initial data, we explicitly find bulges on the $x=0$ axis at times $t \approx \pm 0.06$.}\label{fig-Phi11}
\end{center}
\end{figure}

The result is shown in Fig.~\ref{fig-Phi11} where at $x=0$ and $t\approx 0.06$ we have bulge extremal surfaces. We can confirm that the surface is a bulge because $\partial_t^2 \phi <0$ and $\partial_x^2 \phi<0$ at the surface. We schematically depict the entire solution in Fig.~\ref{fig-Phi11b}.

\begin{figure}[ht]
\begin{center}
\includegraphics[width=0.5\textwidth]{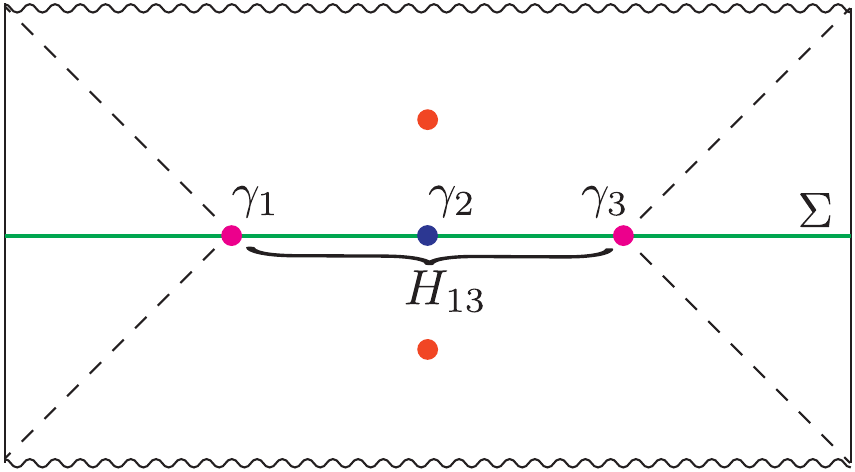}
\caption{A Penrose diagram of our JT gravity plus massive scalar solution with timelike separated bulges. The precise location of the singularity (which in JT gravity plus matter is just a region where the dilaton diverges to negative infinity) is shown schematically and not derived explicitly. A Cauchy slice $\Sigma$ with time-symmetric data is prepared containing a bounce ($\gamma_2$), two throats ($\gamma_1$ and $\gamma_3$), and two asymptotic AdS regions. By maximinimax and time-symmetry, $D(H_{13})$ contains at least two bulges (shown in orange) that are time-like separated from each other.}\label{fig-Phi11b}
\end{center}
\end{figure}

\subsection{Timelike-Separated Throats}\label{sec:3-2}

The strategy here is again to provide complete initial data on a time-symmetric and reflection symmetric (across $x=0$) slice $\Sigma$ with two asymptotically AdS$_2$ boundary conditions of the fields. The implementation in this case is more brute-force. We will construct a time interval data, on $\{x=0, -t_0<t<t_0\}$, such that at the boundaries of the interval we have throat extremal surfaces. By time symmetry, this means at $t=x=0$ we will have another extremal surface which is a bounce. We will then take advantage of the simplicity of the two-dimensional setup to evolve this data sideways to construct a domain of dependence with a time-symmetric spatial slice $\Sigma'$ (see Fig.~\ref{fig:sliceschematic}). The rest of the task is then to complete $\Sigma'$ into full initial data on both sides with AdS$_2$ asymptotics. We do not  generate the complete Cauchy evolution of this initial data; nevertheless we may still immediately make some conclusions about the location of other extremal surfaces. Specifically, as illustrated in Fig.~\ref{fig:sliceschematic}, there must exist four bulges away from the $x=0$ slice as a consequence of maximinimax between one of the timelike-separated throats and one of the outer throats.

Explicitly, we choose the following profile for the time interval:
\begin{align}
\psi(t) = \psi_0 \left(2+t -\frac{2}{a} \log \left(\frac{1+e^{a t}}{2}\right)\right)
\end{align}
with a=0.5 and further
\begin{align}
\begin{cases}
\kappa=0.01,\\
m=10,\\
\psi_0=50.
    \end{cases}
\end{align}
We then demand that $\phi$ solves the $xx$ component of its equation of motion on the time interval:
\begin{align}
- \phi(t) - \partial_t^2 \phi(t) = \frac{1}{2}\kappa \left(-m^2  \psi(t)^2 + (\partial_t\psi(t))^2\right)
\end{align}
such that $\partial_t\phi(t,x=0)\rvert_{t=t_0}=\partial_t\phi(t,x=0)\rvert_{t=0}=0$ for $t_0 = 0.15$. The resulting $\phi(t,x=0)$ profile is depicted in Fig.~\ref{fig-timelikethroat1}. We can directly compute $\partial_x^2 \phi$ at $t=t_0$ using the dilaton equations of motion and confirm that the extremal surfaces at $x=0$, $t=\pm t_0$ are indeed throats, whereas the positive $\partial_t^2 \phi \rvert_{t=x=0}$ implies that the $t=x=0$ extremal surface is a bounce. We then numerically evolve this data ``sideways''to derive the spatial initial data on $\Sigma'$ from which evolves to this solution. The $\phi$ and $\psi$ profiles are plotted in Fig.~\ref{fig-timelikethroat2}.

Next, to complete $\Sigma'$, to a full initial data set with AdS$_2$ asymptotics, we extend the $\psi$ data with a linearly decreasing profile from the boundary of $\Sigma'$ until we reach zero past which we fully turn off $\psi$. The constrains Eq. \eqref{eq-tt} then fixes the corresponding $\phi$. Fig.~\ref{fig-timelikethroat2} depicts the corresponding profiles on the $t=0$ slice. The linear growth of $\phi$ is the asympotic region guarantees AdS$_2$ asymptotics. 

\begin{figure}[ht]
\begin{center}
\includegraphics[width=0.5\textwidth]{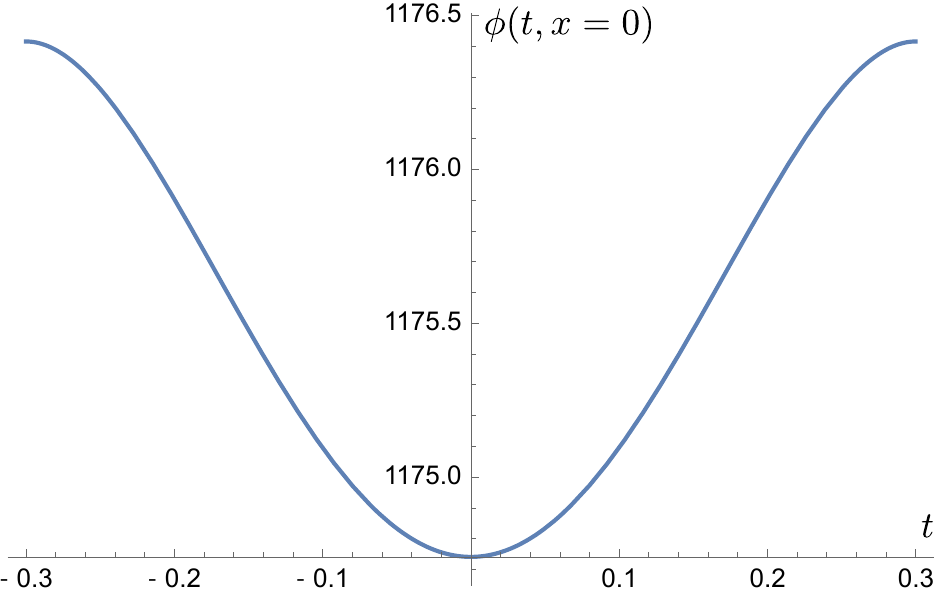}
\caption{The profile of $\phi$ on the $x=0$ axis with $t \in [-t_0,t_0]$ shows three extremal surfaces, at $t=-t_0,t=0$, and $t=t_0$.}\label{fig-timelikethroat1}
\end{center}
\end{figure}

\begin{figure}[ht]
\begin{center}
\includegraphics[width=0.5\textwidth]{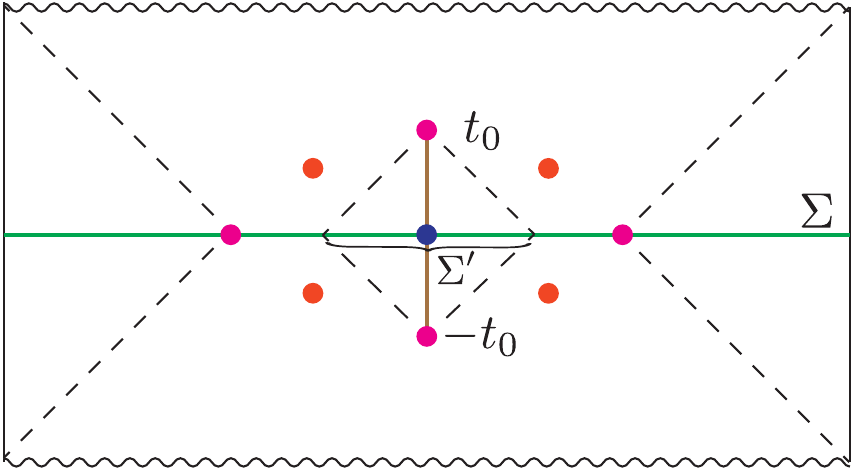}
\caption{A schematic Penrose diagram of our JT gravity plus massive scalar solution with timelike separated throats. First, we prepare data on the timelike interval from $-t_0$ to $t_0$ (shown in brown), with two throats on each end and a bounce (shown in blue) in the middle. The data is then evolved side ways to obtain data on $\Sigma'$, which we then complete to a full Cauchy slice $\Sigma$ with two asymptotic AdS boundary. Then, it is manifest that the spacetime arising from $\Sigma$ contains two time-like separated throats. In addition, there are four bulges (shown in orange) as a consequence of a maximinimax procedure between the throats at $t=\pm t_0$ and ones at $t=0$. }\label{fig:sliceschematic}
\end{center}
\end{figure}

\begin{figure}[ht]
\begin{center}
\includegraphics[width=0.6\textwidth]{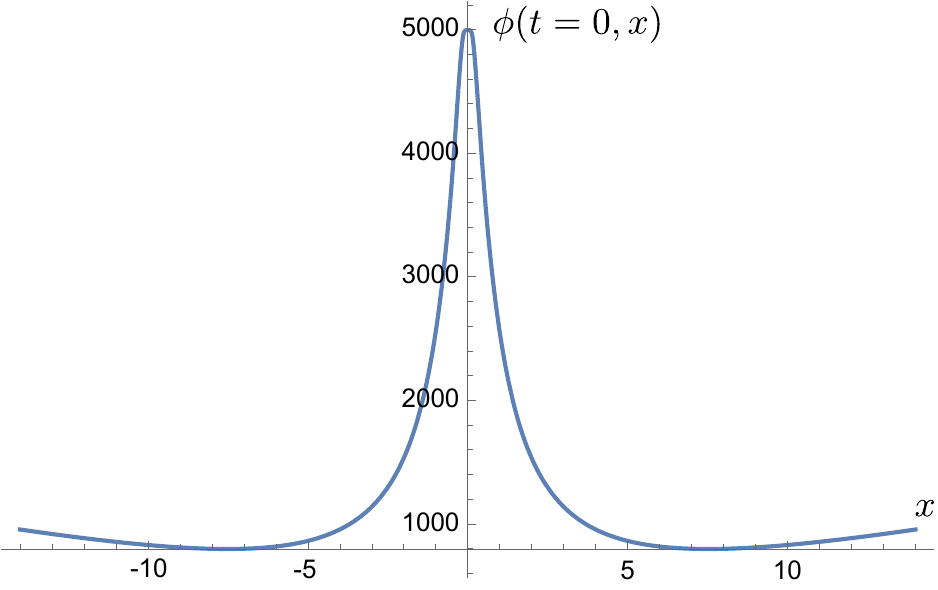}
\caption{The profile of $\phi$ at $t=0$ shows three classical extremal surface at $x=0$ and $x \approx \pm 7$ connected to an asymptotic region. The asymptotic region is characterized by a linearly growing $\phi$.}\label{fig-timelikethroat2}
\end{center}
\end{figure}

\subsection{Remarks}

Given the above examples, timelike-separated QESs of different types are allowed, and we find it highly plausible timelike-separation between any kind of QESs can be realized. An intriguing feature of the above solutions is the presence of a bounce in the region $I^+(\gamma_1) \cap I^-(\gamma_2)$ where $\gamma_1$ and $\gamma_2$ are the early and late throat/bulge surfaces respectively. In fact, in the classical limit with spherical symmetry it is easy to show that $I^+(\gamma_1) \cap I^-(\gamma_2)$, if non-empty, always contains a bounce if $\gamma_1$ and $\gamma_2$ are bulges or throats (not necessarily the same type). Since the problem is effectively 2d in this case, we can simply reverse time and space in the region between $\gamma_1$ and $\gamma_2$ and apply a standard restricted maximin procedure. The outcome will be a bounce extremal surface (in the original spacetime). It is tempting to speculate whether more generally, time-like separated bulges and/or throats imply the existence of a bounce QES, but we leave this to future work.

It would be interesting to find the boundary states dual to the bulk geometries presented here. Since our constructions are Lorentzian, they do not explicitly reveal asymptotic boundary conditions (Euclidean or complex) which would result in these bulk solutions. When such boundary conditions are known, they provide an explicit realization of the boundary state dual to the bulk solution. Various such solutions are discussed in the literature (see~\cite{Goel:2018ubv} for example), but to our knowledge these solutions only contain spacelike separated bulges and throats.

\section{A Refined Python's Lunch Proposal}\label{sec:proposal}

Having established the existence of timelike-separated QESs, we need to understand their role in controlling information flow in the bulk-to-boundary map. The QES prescription~\cite{EngWal14} and entanglement wedge reconstruction~\cite{DonHar16} as usually stated can accommodate timelike-separated QESs without any issue. However, the same cannot be said for the Python's lunch conjecture.

In its original formulation~\cite{BroGha19}, the Python's lunch conjecture assumed an entanglement wedge that contained only two throats. These two throats are necessarily a) the minimal QES $\gamma_{\rm min}$ (for consistency with later notation we shall also call this the dessert surface $\gamma_{\rm dessert}$) and b) the outermost QES (we shall call this the appetizer surface $\gamma_{\rm aptz}$). There always exists at least one additional QES in the entanglement wedge that is not a throat, namely  $\mathrm{maximinimax}(\gamma_{\rm min}, \gamma_{\rm aptz})$. We shall call this surface $\gamma_\mathrm{main}$ to complete the culinary theme; in previous work it was normally called the bulge surface. The wedge $\mathcal{W}_O[\gamma_{\rm dessert}] \cap \mathcal{W}'_O[\gamma_{\rm aptz}]$ is the eponymous Python's lunch.

\paragraph{The Python's Lunch conjecture for spacetimes with two throats \cite{BroGha19}:} the complexity $C$ of reconstructing bulk operators in the lunch $\mathcal{W}_O[\gamma_{\rm dessert}] \cap \mathcal{W}'_O[\gamma_{\rm aptz}]$ using only the boundary state on $B$ is\footnote{Here, and in all other versions of the Python's lunch conjecture, we are assuming that the spacetime volume/action is $O(1)$. Wormholes with parametrically large volume, for example, have a complexity that grows linearly with their volume, even in the absence of a Python's lunch~\cite{Stanford:2014jda, Susskind:2018pmk}. However unless the wormhole is exponentially long, such effects are subleading compared to the contribution to the reconstruction complexity from a Python's lunch.}
\begin{align} 
    \log C = \frac{1}{2} (S_\text{gen}(\gamma_{\rm main}) - S_\text{gen}(\gamma_{\rm aptz})) + O(1).
\end{align}
\vspace{0.5cm}

This conjecture was generalized in~\cite{EngPen21b} to situations with a set of nested, spacelike-separated bulges and throats, with a bulge sandwiched between each pair of throats and vice versa. To state this version of the conjecture, we first recall from Section \ref{sec:outerminimal} that operators outside of an outer-minimal QES $\gamma_{\rm dessert}$ can be reconstructed by someone knowing only about $W_{O}[\gamma_{\rm dessert}]$. This suggests that only QESs contained within $W_{O}[\gamma_{\rm dessert}]$ should be relevant to the reconstruction complexity of operators in $W_{O}[\gamma_{\rm dessert}]$. Since all the QESs are assumed to be nested, we can label the extremal surfaces within $W_{O}[\gamma_{\rm dessert}: = \gamma_0]$ by $\gamma_i$ for $i> 0$ such that $$W_{O}[\gamma_{i}] \subseteq W_{O}[\gamma_{j}]$$ whenever $i > j$. For this version of the conjecture, we also assume that there exists a single Cauchy slice $\Sigma$, containing all the extremal surfaces $\gamma_i$, such that $\Sigma$ is simultaneously a maximal Cauchy slice for all relevant maximin or maximinimax constructions. The slice $\Sigma$ can be thought of as the gravitational analogue of a tensor network. (As with all other assumptions, the assumption of the existence of the slice $\Sigma$ will be removed in the fully general prescription we propose in Section \ref{sec:newconj}.)

\paragraph{The Python's Lunch conjecture for multiple spacelike-separated extremal surfaces \cite{EngPen21b}:} the restricted complexity $C$ of decoding bulk operators that lie outside the outer-minimal QES $\gamma_{\rm dessert}$ (but not outside any outer-minimal QES $\gamma_j$ for $j > 0$) satisfies
\begin{align} \label{eq:multilunch}
    \log C = \max_{j > i:} \left [\frac{1}{2} (S_\text{gen}(\gamma_i) - S_\text{gen}(\gamma_j))\right] + O(1)= \frac{1}{2} (S_\text{gen}(\gamma_{\rm main}) - S_\text{gen}(\gamma_{\rm aptz})) + O(1),
\end{align}
where $\gamma_{\rm main}$ and $\gamma_{\rm aptz}$ are defined as the surfaces in which the maximum is attained.
\vspace{0.5cm}

Since a bulge always has larger generalized entropy than the neighboring throats, and a throat smaller generalized entropy than the neighbouring bulges, the surface $\gamma_{\rm main}$ is always a bulge and the surface $\gamma_{\rm aptz}$ is always a throat. In fact, if we know one of $\gamma_{\rm aptz}$ and $\gamma_{\rm main}$, the other can be found by a maximinimax or maximin prescription respectively:

\begin{lem}
It is always the case that $\gamma_{\rm aptz} = \mathrm{maximin}(\gamma_\mathrm{main}, B)$.
\end{lem}
\begin{proof}
By standard arguments, the maximin surface is the QES in the exterior of $\gamma_{\rm main}$ with smallest generalized entropy. But this is exactly the surface $\gamma_{\rm aptz}$ found by maximizing over $j$ in \eqref{eq:multilunch}.
\end{proof}

\begin{lem}
It is always the case that $\gamma_{\rm main} = \mathrm{maximinimax}(\gamma_{\rm dessert},\gamma_{\rm aptz})$.
\end{lem}
\begin{proof}
Since the maximinimax surface is always a bulge QES, it must trivially be $\gamma_k$ for some odd $k$ with $0 < k < j$. To prove that it is the surface $\gamma_i$ in that range with largest generalized entropy, we show by induction that (subject to the asssumptions above) the maximinimax surface between two throats $\gamma_a$ and $\gamma_b$  with $a < b$ is the intermediate QES $\gamma_c$ with largest generalized entropy. The base case where $b = a+2$ is trivial. Now assume that the result holds whenever $b - a < n$. For $b - a = n$, let $\gamma_d$ be the maximinimax surface and let $S_\mathrm{gen}(\gamma_d) < S_\mathrm{gen}(\gamma_c)$. By the symmetry of the problem, it is sufficient to check the case where $c < d$. 

Let $\gamma_e$ be the smallest generalized entropy QES with $c < e < d$. By construction, there exists a foliation $\gamma(t)$ of the Cauchy slice $\Sigma$ defined above from $\gamma_a$ to $\gamma_d$ with $S_\mathrm{gen}(\gamma(t)) \leq S_\mathrm{gen}(\gamma_d)$ everywhere. We define a new foliation $\tilde\gamma(t)$ from $\gamma_a$ to $\gamma_e$ by $W_{O}[\tilde\gamma(t)] = (W_{O}[\tilde\gamma(t)]' \cap W_{0}[\gamma_e]')'$. Similarly, the foliation $\dbtilde{\gamma}(t)$ from $\gamma_e$ to $\gamma_d$ is defined by $W_{O}[\dbtilde{\gamma}(t)] = W_{O}[\dbtilde{\gamma}(t)] \cap W_{0}[\gamma_e]$. 

By strong subadditivity, we have
\begin{align}
    S_\mathrm{gen}[\gamma(t)] + S_\mathrm{gen}[\gamma_e] \geq S_\mathrm{gen}[\tilde\gamma(t)] + S_\mathrm{gen}[\dbtilde{\gamma}(t)].
\end{align}
But, by assumption, $\gamma_e$ is minimal in $\Sigma$ between $\gamma_e$ and $\gamma_d$ and so $S_\mathrm{gen}[\dbtilde{\gamma}(t)] \geq S_\mathrm{gen}[\gamma_e]$. It follows that for all $t$, we have
\begin{align}
    S_\mathrm{gen}[\tilde\gamma(t)] \leq S_\mathrm{gen}[\gamma(t)] \leq S_\mathrm{gen}[\gamma_d].
\end{align}
Since by assumption $\Sigma$ is a maximinimax slice between $\gamma_a$ and $\gamma_e$, the existence of the foliation $\tilde\gamma(t)$ means that $\gamma_c$ cannot be maximinimax between $\gamma_a$ and $\gamma_e$. But this gives us our desired contradiction since $e - a < n$.
\end{proof}

\subsection{A conjecture} \label{sec:newconj}
The original Python's lunch conjecture~\cite{BroGha19} and the generalization from \cite{EngPen21b} were both primarily justified by an appeal to tensor network toy models. Such models, however, are only analogous to a single Cauchy slice of a holographic spacetime and hence cannot contain any analogue of two timelike-separated extremal surfaces. If we want to construct a completely general version of the Python's lunch conjecture in light of our results in the previous section, the only remaining tools that we have to guide us are general covariance and consistency with the previous conjectures in their domains of validity. There does appear to be one particularly natural conjecture consistent with these requirements, which we describe below. As we shall see, it has the nice properties that $\log C$ continues to be determined by the difference between the generalized entropies of a bulge $\gamma_{\rm main}$ and a throat $\gamma_{\rm aptz}$, and that $\gamma_{\rm aptz}$ can be found from $\gamma_{\rm main}$, and $\gamma_{\rm main}$ from $\gamma_{\rm aptz}$, by maximin and maximinimax prescriptions respectively.

\paragraph{A general Python's Lunch conjecture:} the restricted complexity $C$ of decoding a bulk operator at a point $p$ is given by:
\begin{align} \label{eq:PLtotallgeneral}
\log C & =\min_{\gamma_0} \max_{\gamma_1} \left[S_\mathrm{gen}(\mathrm{maximinimax}(\gamma_0, \gamma_1)) - S_\mathrm{gen}(\gamma_1)\right] + {\cal O}(1)\\ &=\left[S_\mathrm{gen}(\gamma_{\rm main}) - S_\mathrm{gen}(\gamma_{\rm aptz})\right] + {\cal O}(1).
\end{align}
Here the minimization is over outer-minimal QESs $\gamma_0$ such that $p \in W_{O}[\gamma_{0}]$, and the minimizing surface is $\gamma_{\rm dessert}$. The maximization is over QESs $\gamma_\mathrm{1} \subseteq W_{O}[\gamma_{0}]$, and the maximizing surface is $\gamma_{\rm aptz}$. Finally the surface $\gamma_\mathrm{main} = \mathrm{maximinimax}(\gamma_{\rm dessert}, \gamma_\mathrm{aptz})$.
\vspace{0.5cm}

By construction, the surface $\gamma_\mathrm{main}$ is maximinimax between $\gamma_{\rm dessert}$ and $\gamma_\mathrm{aptz}$. What is not immediately obvious, but is vitally important for consistency with the intuition of tensor networks, is that $\gamma_\mathrm{aptz}$ is maximin in the exterior of $\gamma_{\rm main}$. In fact the proof that this is indeed always the case is somewhat involved, so we will break it down into a few lemmas.

\begin{lem}
There does not exist a QES $\gamma_\mathrm{ext} \subseteq W_O[\gamma_\mathrm{aptz}]$ such that $S_\mathrm{gen}(\gamma_\mathrm{ext}) \leq S_\mathrm{gen}(\gamma_\mathrm{aptz})$.
\end{lem}
\begin{proof}
Suppose there exists such a QES $\gamma_\mathrm{ext}$. Without loss of generality we may assume that $\gamma_\mathrm{ext}$ is maximin between $\gamma_\mathrm{aptz}$ and $B$. To show a contradiction, we need to show that the size of the lunch between $\gamma_{\rm dessert}$ and $\gamma_\mathrm{ext}$ is larger than that between $\gamma_{\rm dessert}$ and $\gamma_\mathrm{ext}$ i.e. that
\begin{align} \label{eq:desiredbound}
    S_\mathrm{gen}(\mathrm{maximinimax}(\gamma_{\rm dessert}, &\gamma_\mathrm{ext})) - S_\mathrm{gen}(\gamma_\mathrm{ext}) \stackrel{?}{>} \nonumber\\&S_\mathrm{gen}(\mathrm{maximinimax}(\gamma_{\rm dessert}, \gamma_\mathrm{aptz})) - S_\mathrm{gen}(\gamma_\mathrm{aptz}).
\end{align}
To do so we pick a particular Cauchy slice $\Sigma$ and choose a minmax foliation $\gamma(t)$ of that slice between $\gamma_{\rm dessert}$ and $\gamma_\mathrm{ext}$. For any $t$, we must have
\begin{align} \label{eq:gammabound}
   S_\mathrm{gen}(\gamma(t)) \leq S_\mathrm{gen}(\mathrm{maximinimax}(\gamma_{\rm dessert}, \gamma_\mathrm{ext})).
\end{align}
Let us take $\Sigma$ to be the union of a maximinimax partial Cauchy slice between $\gamma_{\rm dessert}$ and $\gamma_\mathrm{aptz}$ and a partial Cauchy slice between $\gamma_\mathrm{aptz}$ and $\gamma_\mathrm{ext}$ on which $\gamma_\mathrm{ext}$ is minimal. We then define two new foliations $\tilde\gamma(t)$ and $\dbtilde\gamma(t)$ by $W_O[\tilde\gamma(t)] = (W_O[\gamma(t)]' \cap W_O[\gamma_\mathrm{aptz}]')'$ and $W_O[\dbtilde\gamma(t)] = W_O[\gamma(t)] \cap W_O[\gamma_\mathrm{aptz}]$ respectively. The surfaces $\tilde\gamma(t)$ foliate between $\gamma_{\rm dessert}$ and $\gamma_\mathrm{aptz}$ in a maximinimax slice. Hence there must exist $t$ with
\begin{align} \label{eq:gamma'bound}
    S_\mathrm{gen}(\tilde \gamma(t)) \geq S_\mathrm{gen}(\mathrm{maximinimax}(\gamma_{\rm dessert}, \gamma_\mathrm{aptz}).
\end{align}
On the other hand, for all $t$, $\dbtilde\gamma(t)$ is contained in a partial Cauchy slice within which $\gamma_\mathrm{ext}$ has minimal generalized entropy. So
\begin{align}\label{eq:gamma''bound}
    S_\mathrm{gen}(\dbtilde \gamma(t)) \geq S_\mathrm{gen}(\gamma_\mathrm{ext}),
\end{align}
with equality only when $\gamma(t) = \gamma_{\rm ext}$. But by strong subadditivity we have for all $t$
\begin{align} \label{eq:ssagammas}
    S_\mathrm{gen}(\gamma(t)) + S_\mathrm{gen}(\gamma_\mathrm{aptz}) \geq S_\mathrm{gen}(\tilde\gamma(t)) + S_\mathrm{gen}(\dbtilde\gamma(t)).
\end{align}
Substituting the inequalities \eqref{eq:gammabound}, \eqref{eq:gamma'bound} and \eqref{eq:gamma''bound} into \eqref{eq:ssagammas} gives exactly the inequality \eqref{eq:desiredbound} that we needed to show to prove our desired contradiction.
\end{proof}

\begin{lem}
There does not exist a QES $\gamma_\mathrm{ext} \subseteq W_O[\gamma_\mathrm{aptz}]' \cap W_O[\gamma_{\rm dessert}]$ such that $S_\mathrm{gen}(\gamma_\mathrm{ext}) < S_\mathrm{gen}(\gamma_\mathrm{aptz})$.
\end{lem}
\begin{proof}
We construct a Cauchy slice $\Sigma$ for $W_O[\gamma_\mathrm{aptz}]' \cap W_O[\gamma_{\rm dessert}]$ that maximizes
\begin{align}
    \min_{\gamma(t)} \left[A S_\mathrm{gen}(\gamma(t_\mathrm{max})) + B S_\mathrm{gen}(\gamma_\mathrm{min})\right].
\end{align}
Here $A$ and $B$ are positive constants, $t_\mathrm{max}$ maximizes $S_\mathrm{gen}(\gamma(t))$ and $\gamma_\mathrm{min} \subseteq W_O[\gamma(t_\mathrm{max})\cap \Sigma$ has minimal $S_\mathrm{gen}$ among all surfaces in $\Sigma$ in the exterior of $\gamma(t_\mathrm{max})$.

It follows immediately from the definition that $\gamma_\mathrm{min} = \mathrm{maximin}(\gamma(t_\mathrm{max}), \gamma_{\rm aptz})$. Moreover, when the ratio $A/B$ is sufficiently large, the surface $\gamma(t_\mathrm{max})$ must approach the maximinimax surface $\gamma_{\rm main}$. By standard arguments, since $\gamma(t_\mathrm{max})$ and $\gamma_\mathrm{aptz}$ are extremal, the surface $\gamma_\mathrm{min}$ is the smallest extremal surface between them. Hence $\gamma_\mathrm{min}$ and $\gamma_\mathrm{aptz}$ are distinct and $S_\mathrm{gen}(\gamma_\mathrm{min}) < S_\mathrm{gen}(\gamma_\mathrm{aptz})$ if and only if there exists a QES $\gamma_\mathrm{ext}$ satisfying the conditions above. We assume this is the case and prove a contradiction .

Let $\gamma(t)$ be a minmax foliation of $\Sigma$ (maximized at $t_{\rm max}$) between $\gamma_{\rm dessert}$ and $\gamma_\mathrm{aptz}$ and let $\tilde\gamma(t)$ and $\dbtilde\gamma(t)$ be defined respectively by $W_O[\tilde\gamma(t)] = (W_O[\gamma(t)]' \cap W_O[\gamma_\mathrm{min}]')'$ and $W_O[\dbtilde\gamma(t)] = W_O[\gamma(t)] \cap W_O[\gamma_\mathrm{min}]$. For $t \leq t_\mathrm{max}$, we have $\dbtilde\gamma(t) = \gamma_\mathrm{min}$, and hence $S_\mathrm{gen}(\dbtilde\gamma(t)) < S_\mathrm{gen}(\gamma_{\rm main})$. For $t > t_\mathrm{max}$, we have
\begin{align}
    S_\mathrm{gen}(\gamma(t)) + S_\mathrm{gen}(\gamma_\mathrm{min}) \geq S_\mathrm{gen}(\tilde\gamma(t)) + S_\mathrm{gen}(\dbtilde\gamma(t))
\end{align}
by strong subadditivity. Since $S_\mathrm{gen}(\tilde\gamma(t)) \geq S_\mathrm{gen}(\gamma_\mathrm{min})$ and $S_\mathrm{gen}(\gamma(t)) <S_\mathrm{gen}(\gamma_{\rm main})$ for $t > t_\mathrm{max}$, we again obtain $S_\mathrm{gen}(\dbtilde\gamma(t)) < S_\mathrm{gen}(\gamma_{\rm main})$.

Now suppose the minmax surface between $\gamma_{\rm dessert}$ and $\gamma_\mathrm{min}$ in $\Sigma$ had generalized entropy strictly smaller than $S_\mathrm{gen}(\gamma_{\rm main})$. We could then combine a minmax foliation between $\gamma_{\rm dessert}$ and $\gamma_\mathrm{min}$ with $\dbtilde\gamma(t)$ to obtain a foliation of $\Sigma$ between $\gamma_{\rm dessert}$ and $\gamma_\mathrm{aptz}$ with maximal generalized entropy strictly smaller than $S_\mathrm{gen}(\gamma_{\rm main})$, contradicting the definition of $\gamma_{\rm main}$. But, on the other hand, if the minmax surface between $\gamma_{\rm dessert}$ and $\gamma_\mathrm{min}$ has generalized entropy at least as large as $S_\mathrm{gen}(\gamma_{\rm main})$, then 
\begin{align}
S_\mathrm{gen}(\mathrm{maximinimax}(&\gamma_{\rm dessert}, \gamma_\mathrm{min})) - S_\mathrm{gen}(\gamma_\mathrm{min}) > \\&S_\mathrm{gen}(\mathrm{maximinimax}(\gamma_{\rm dessert}, \gamma_\mathrm{aptz})) - S_\mathrm{gen}(\gamma_\mathrm{aptz}),
\end{align}
which contradicts the maximality of $\gamma_\mathrm{aptz}$ in \eqref{eq:PLtotallgeneral} if $S_\mathrm{gen}(\gamma_\mathrm{min}) < S_\mathrm{gen}(\gamma_\mathrm{aptz})$. 
\end{proof}

\begin{thm}
The surface $\gamma_\mathrm{aptz}$ is maximin between $\gamma_{\rm main}$ and $B$.
\end{thm}
\begin{proof}
The surface $\gamma_\mathrm{aptz}$ is maximin if there exists no smaller QES in $W_O[\gamma_{\rm main}]$. By Lemma \ref{lem:ssa}, if there exists a QES $\gamma_\mathrm{ext} \subset W_O[\gamma_{\rm main}]$ with $S_\mathrm{gen}(\gamma_\mathrm{ext}) < S_\mathrm{gen}(\gamma_\mathrm{aptz})$ then there must also exist a QES $\gamma_\mathrm{ext}'$ satisfying the same condition with either $\gamma_\mathrm{ext}' \subseteq W_O[\gamma_\mathrm{aptz}]$ or $\gamma_\mathrm{aptz} \subseteq W_O[\gamma_\mathrm{ext}']$. Both possibilities are ruled out by the preceding lemmas.
\end{proof}

\begin{cor}
The surface $\gamma_\mathrm{aptz}$ is outer minimal (and therefore a throat).
\end{cor}

\begin{cor}
The fully general Python's lunch prescription reduces to the prescription for multiple spacelike-separated extremal surfaces when all QESs are spacelike-separated and nested and a Cauchy slice $\Sigma$ with appropriate properties exists.
\end{cor}

Since the true Python's lunch requires a somewhat complicated procedure to compute, the following simpler definition can be helpful

\begin{defn}
    The \emph{full-course Python's lunch} is defined by setting $\gamma_{\rm dessert}$ to be the minimal QES $\gamma_\mathrm{min}$ and $\gamma_\mathrm{aptz}$ to be the outermost QES $\gamma_\mathrm{outer}$. Its size is given by
    \begin{align}
        S_\mathrm{gen}(\gamma_{\rm main}) - S_\mathrm{gen}(\gamma_\mathrm{outer}),
    \end{align}
    where $\gamma_{\rm main} = \mathrm{maximinimax}(\gamma_\mathrm{min},\gamma_\mathrm{outer})$.
\end{defn}
\begin{cor}
    Assuming the minimal QES is isolated, the size of the Python's lunch for a point $p$ in a sufficiently small neighbourhood of $\gamma_\mathrm{min}$ (but within the entanglement wedge) is lower bounded by the size of the full-course Python's lunch. 
\end{cor}

\subsection{Where is the tensor network?}
The general Python's Lunch proposal above is satisfyingly clean and has a number of additional nice features as we have just described. However it also leads to some somewhat surprising conclusions about the relationship between tensor network (TN) models and quantum gravity. 
In general, the holographic duality requires the existence of a linear map from the bulk Hilbert space to the boundary Hilbert space. Though this map in general has not yet been derived ab initio, TN toy models have very successfully reproduced features of the map that are expected from basic gravitational path integral computations, e.g., the path integral derivation of the QES prescription~\cite{LewMal13, FauLew13, DonLew17, AlmHar19, PenShe19}. These setups construct a TN which represents the spatial geometry of a given bulk slice and asymptotes to the boundary slice of choice. For generic time-dependent spacetimes, it is not always obvious which Cauchy slice is meant to be analogous to a TN description of the boundary state. It cannot be any Cauchy slice, at least if we want the logarithm of the dimension of a cut through the tensor network to be proportional to the generalized entropy of the corresponding surface. In fact, consistency with the QES prescription requires suggests that the ``TN slice'' be a maximin slice, with the minimal QES the minimal generalized entropy surface within the slice. Consistency with the Python's lunch conjecture  constrains any candidate ``TN slice'' even further, forcing it to include $\gamma_{\rm dessert}$, $\gamma_{\rm main}$ and $\gamma_{\rm aptz}$. Indeed, in a spacetime where all the QESs homologous to a given region $B$ are spacelike, it is natural to think of the TN slice as containing \emph{all} of the QESs, as was assumed in the version of the Python's lunch proposed in \cite{EngPen21b}.

However, in the presence of timelike-separated QESs, no Cauchy slice can contain every QES of a given boundary (or boundary region). Our fully general Python's lunch conjecture suggests that a hypothetical TN slice should definitely contain $\gamma_\mathrm{dessert}$, $\gamma_{\rm main}$ and $\gamma_{\rm aptz}$, along with the minimal QES. Since all of these extremal surfaces are spacelike separated, this is certainly possible. However different bulk operators (corresponding to bulk points $p$) will in general have different surfaces $\gamma_\mathrm{dessert}$, $\gamma_{\rm main}$ and $\gamma_{\rm aptz}$. And there is no reason that \emph{all} of those surfaces need to lie within the same Cauchy slice.

Issues with finding a consistent ``tensor network slice'' in time-dependent spacetimes are not new: they show up even when merely demanding consistency with the QES prescription once one considers multiple overlapping boundary regions. However, such issues were previously thought to be avoided in static or time-reflection symmetric spacetimes. In those case, it prima facie appears obvious that ``TN slice'' should be the time-reflection symmetric Cauchy slice. All minimal QESs for all time-reflection symmetric boundary regions are contained in  this slice. And in fact \textit{if} all QESs were supposed to be spacelike-separated, then they would have to be contained in the time-reflection symmetric slice (otherwise they would be timelike-separated from their image under the time reflection symmetry). So, naively, the Python's lunch conjecture for spacelike-separated surface is still completely consistent with a TN describing the time-reflection symmetric slice.

However, here we come up against the examples of Section~\ref{sec:egs}: time-reflection symmetric slices that nevertheless do \textit{not} contain the QESs relevant for the general Python's Lunch conjecture proposed above. We appear to have found a clash between two principles supported by numerous evidence from AdS/CFT: first, the maximin prescription and its various descendants as the identifiers of relevant QESs in generic spacetimes, and two, modeling the bulk-to-boundary map as a tensor network whenever the bulk has time reflection symmetry. If the former is correct, then all of the previous insights about tensor networks in AdS/CFT appear to be suspect -- which is highly surprising given the significant progress derived from them. If the latter is correct, then it is possible that the reconstruction complexity is much smaller than suggested by a maximin-supported proposal: implementing Grover search on the much-smaller bounce in between two bulges will need to undo significantly less postselection.

Given the stark difference between these conclusions, a first principles derivation of bulk reconstruction complexity is desirable. We conclude by offering a concrete approach towards that end, whose implementation is left to future work. In the python's lunch proposal, the source of complexity is the post-selection happening in the TN bulk-to-boundary map as demonstrated by the $S_{\rm gen}$ drop between the relevant bulge and throat in the geometry. Therefore, one way to check the python's lunch proposal is via a direct gravitational computation of the correct amount of post-selection in the bulk-to-boundary map. Indeed, this post-selection imprints itself in the small overlaps in the inner product of boundary states dual to orthogonal bulk states. One can calculate such overlaps (or, more precisely, their norm-squared average) using the gravitational path integral. Doing so will bypass the TN picture -- which we argued is not reliable --  and let gravity itself determine the salient features of the bulk-to-boundary map such as its complexity.

\section*{Acknowledgments} It is a pleasure to thank A. Folkestad, D. Harlow, A. Levine, and Z. Yang for valuable discussions. NE is supported in part by NSF grant no. PHY-2011905, by the U.S. Department of Energy under Early Career Award DE-SC0021886, by the John Templeton Foundation via the Black Hole Initiative, by the Sloan Foundation, by the Heising-Simons Foundation, and by funds from the MIT physics department. GP was supported by the University of California, Berkeley; by the Department of Energy through DE-SC0019380 and DE-FOA-0002563; by AFOSR award FA9550-22-1-0098; and by an IBM Einstein Fellowship at the Institute for Advanced Study. ASM is supported by the National Science Foundation under Award Number 2014215.

\appendix

\section{Obtaining the solution of Sec.~\ref{sec:3-1} using Green's function}\label{app-1}

This is particularly simple for the scalar field using AdS$_2$ Greens functions. here we use AdS$_2$ coordinates
\begin{align}
ds^2 = \frac{1}{\cos^2 \rho}(-d\tau^2 + d\rho^2),
\end{align}
where $t=\tau$ and $x= \tan(\rho)$.

We will find the retarded Green's function:
\begin{align}\label{eq-greendelta}
(\Box - m^2) G_R(x,x') = \frac{1}{\sqrt{-g}}\delta^{(2)}(x-x')
\end{align}
Using the following ansatz:
\begin{align}
    G_R(x,x') = f_R(\sigma(x;x'))\theta(\tau-\tau')\theta(\tau-\tau' - | \rho-\rho'|)
\end{align}
with:
\begin{align}
    \sigma(\tau,\rho;\tau',\rho') = \frac{\cos(\tau-\tau')-\sin\rho \sin \rho'}{\cos \rho \cos \rho'}
\end{align}
Plugging into Eq. \eqref{eq-greendelta}, we get:
\begin{align}
    (\sigma^2-1)f_R''+2\sigma f_R'-\Delta(\Delta-1) f_R=0
\end{align}
where $\Delta = 1/2 + \sqrt{1/4+m^2 R^2}$. The solution is given by a Legendre function. Putting everything together we get:
\begin{align}
    G_R = -\frac{1}{2} P_{\Delta-1}(\sigma)\theta(\tau-\tau')\theta(\tau-\tau' - | \rho-\rho'|)
\end{align}

Given some initial time symmetric conditions $\psi_0(\tau=0,\rho)$ and $\partial_\tau\psi_0(\tau,\rho)\rvert_{\tau=0}=0$, for  $\tau >0$ we have:
\begin{align}
\psi (\tau,\rho=0) = -\frac{1}{2} \int_{-\infty}^{\infty} d\rho'~\psi_0(0,\rho') \left(P'_{\Delta-1}(\sigma) \frac{\sin(\tau)}{\cos\rho \cos\rho'} \theta(\tau-|\rho-\rho'|)+ P_{\Delta-1}(\sigma) \delta (\tau-|\rho-\rho'|)\right)
\end{align}
and
\begin{align}
    \sigma(\tau,\rho;\tau',\rho') = \frac{\cos(\tau-\tau')-\sin\rho \sin \rho'}{\cos \rho \cos \rho'}
\end{align}
We can now use the following dilaton equation of motion to find the dilaton profile along the $\rho=0$ line.
\begin{align}
\phi(\tau, \rho=0)+\partial_\tau^2 \phi(\tau, \rho=0) = \frac{\kappa}{2} \left(m^2 \psi(\tau, \rho=0)^2 - (\partial_\tau \psi(\tau, \rho=0))^2 \right)
\end{align}

The resulting value of $\phi(t,x=0)$ is plotted in Fig.~\ref{fig-analyticphi} below, in agreement with Fig.~\ref{fig-Phi11}.

\begin{figure}
\begin{center}
\includegraphics[width=0.5\textwidth]{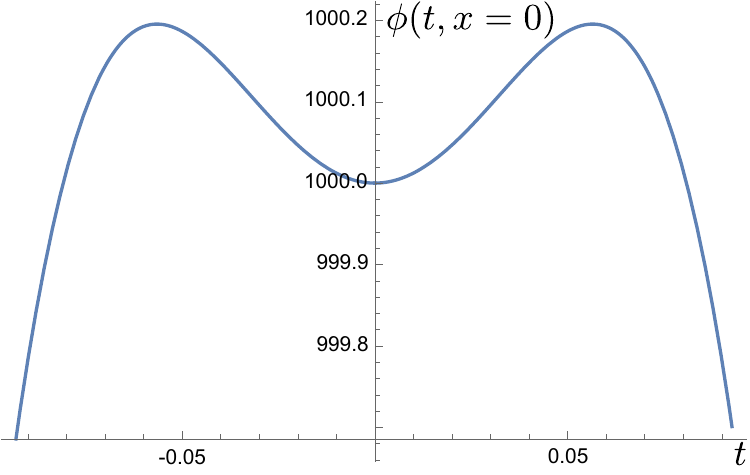}
\caption{Analytic solution to the dilaton profile on the $x=0$ axis for the construction of Sec.~\ref{sec:3-1}.}\label{fig-analyticphi}
\end{center}
\end{figure}

\bibliographystyle{jhep}
\bibliography{all}

\end{document}